\documentclass[a4paper,twoside,12pt]{article}
\usepackage{pict2e}
\usepackage{graphicx}
\usepackage{color}
\usepackage{dsfont}

\makeatletter
\@addtoreset{equation}{section}
\@addtoreset{figure}{section}
\makeatother

\newtheorem{prop}{\bf Proposition}[section]
\newtheorem{lem}[prop]{\bf Lemma}
\newtheorem{cor}[prop]{\bf Corollary}

\newtheorem{thm}[prop]{\bf Theorem}
\newtheorem{rem}[prop]{\bf Remark}
\newenvironment{proof}
  {\begin{trivlist}\item[]{\bf Proof.}}
  {\hspace*{\fill}{$\bowtie$}\end{trivlist}}

\newcommand{\setR}{\mathds{R}}
\newcommand{\setZ}{\mathds{Z}}
\newcommand{\setN}{\mathds{N}}

\newcommand{\diff}{\mathrm{\,d}}

\begin{document}
\pagestyle{empty}
\markboth
  {S.~Liebscher, J.~H{\"a}rterich, K.~Webster, M.~Georgi}
  {Ancient Dynamics in Bianchi Models: Approach to Periodic Cycles}
\pagenumbering{arabic}

\vspace*{\fill}

\begin{center}
\Large\bfseries
Ancient Dynamics in Bianchi Models: \\
Approach to Periodic Cycles
\end{center}

\vspace*{\fill}

\begin{center}
   \textbf{\large Stefan Liebscher}
\\ Freie Universit\"at Berlin, Institut f\"ur Mathematik
\\ Arnimallee 3, 14195 Berlin, Germany
\\ \texttt{sliebsch@zedat.fu-berlin.de}
\\ ~
\\ \textbf{\large J\"org H\"arterich}
\\ Ruhr-Universit\"at, Fakult\"at f\"ur Mathematik
\\ Universit\"atsstr.~150, 44780 Bochum, Germany
\\ \texttt{Joerg.Haerterich@ruhr-uni-bochum.de}
\\ ~
\\ \textbf{\large Kevin Webster}
\\ Imperial College London, Department of Mathematics
\\ South Kensington Campus, London SW7 2AZ, UK
\\ \texttt{knwebster@gmail.com}
\\ ~
\\ \textbf{\large Marc Georgi}
\\ Freie Universit\"at Berlin, Institut f\"ur Mathematik
\\ Arnimallee 3, 14195 Berlin, Germany
\\ \texttt{MarcGeorgi@gmx.de}
\end{center}

\vspace*{\fill}

\begin{center}
Preprint, October 2010 
\end{center}

\vspace*{\fill}

\clearpage

\vspace*{\fill}

\begin{abstract}\noindent

We consider cosmological models of Bianchi type. 
In particular, we are interested in the $\alpha$-limit dynamics near the 
Kasner circle of equilibria for Bianchi classes VIII and IX.
They correspond to cosmological models close to the big-bang singularity.

We prove the existence of a codimension-one family of 
solutions that limit, for $t \to -\infty$, onto a heteroclinic 3-cycle to 
the Kasner circle of equilibria.
The theory extends to arbitrary heteroclinic chains that are uniformly bounded away from
the three critical Taub points on the Kasner circle, in particular to all closed heteroclinic
cycles of the Kasner map.

\end{abstract}

\vspace*{\fill}

\cleardoublepage
\setcounter{page}{1}
\pagestyle{myheadings}

\section{Introduction}
\label{secIntroduction}

The Einstein field equations
\[
  R_{\alpha\beta} - \frac{1}{2}Rg_{\alpha\beta} = T_{\alpha\beta}
\]
restricted to spatially homogeneous, non-isotropic space-times $g_{\alpha\beta}$ 
with an ideal non-tilted fluid yield a system of ordinary 
differential equations, the Bianchi class-A model (\ref{eq5dimBianchi}) \cite{WainwrightHsu1989-BianchiA, WainwrightEllis2005-Cosmology}.

The $\alpha$-limit, $t\to-\infty$, of this system corresponds to the dynamics 
near the big-bang singularity. 
The dynamics in this limit, however, is not yet fully understood.

It has been conjectured \cite{Misner1969-Mixmaster, BugalhoEtAl1986-BianchiIX} 
that the dynamics follows the (formal) Kasner map
defined on the Kasner circle of equilibria and given by heteroclinic connections to 
the equilibria on the Kasner circle.
Equilibria on the Kasner circle represent self-similarly expanding space-times.
A trajectory close to a formal heteroclinic sequence thus corresponds to a space-time
that is close to different self-similar spacetimes as it approaches the singularity 
in backward time-direction: a tumbling universe.

At least for Bianchi class-IX solutions the Bianchi attractor formed by the union of the 
Kasner circle and its heteroclinic orbits has been proven to indeed be a (global) attractor 
for trajectories to generic initial data under the time-reversed flow \cite{Ringstroem2001-BianchiIX}.

Rigorous results on the correspondence of the dynamics close to the attractor and the formal sequences of 
heteroclinic orbits to the Kasner circle are still missing. 
They would facilitate the discussion of the Belinskii-Khalatnikov-Lifshitz (BKL) conjecture of spatial decoupling close to singularities 
as well as the Misner hypothesis on the development of (spatially) homogeneous/isotropic space-times 
due to mixing near the initial singularity. See also \cite{HeinzleUggla2009-Mixmaster} for a recent survey.

In this paper we make the first step towards a rigorous description of the $\alpha$-limit 
dynamics of the Bianchi system. We describe the set of initial conditions near the Bianchi 
attractor that follow the (up to equivariance) unique period-3 heteroclinic cycle. 
In fact we prove that this set forms a codimension-one Lipschitz manifold,
see theorem \ref{thStableSet}.

We note that this result does not depend on monotonicity arguments and therefore covers 
both class-VIII and class-IX solutions approaching a period-3 heteroclinic cycle.

We will start by reviewing the Bianchi system in section \ref{secBianchiSystem}. 
in section \ref{secLocalMap} we study the passage near a line of equilibria in a generalized context.
This yields a local map between sections to the period-3 heteroclinic cycle near the Kasner circle.
In section \ref{secReturnMap} this local map is combined with the global excursion given by the heteroclinic cycle.
We obtain a return map with a fixed point representing the heteroclinic cycle.
The stable manifold of this fixed point, 
i.e.~the set of all solutions converging to the fixed point under iterations of the return map, 
represents all solutions following the heteroclinic cycle in the Bianchi system.
We construct this (local) stable manifold as the limit object of a graph transformation.
This yields the claimed Lipschitz set of trajectories with $\alpha$-limit dynamics 
following the period-3 heteroclinic cycle, theorem \ref{thStableSet}. 
Finally, we discuss generalizations to this result as well as open problems in section \ref{secDiscussion}.

Generalizations include heteroclinic cycles of arbitrary period as well as non-periodic heteroclinic 
sequences that do not approach the singular Taub points. 
Matter models between dust and radiation can be included.

Shortly after the submission of this article two more results on the same problem have appeared.
While B{\'e}guin \cite{Beguin2010-BianchiAsymptotics} shows existence of solutions to non-periodic trajectories of the Kasner map 
that remain bounded away from the Taub points like in our treatment, Reiterer $\&$ Trubowitz \cite{ReitererTrubowitz2010-BKL}
construct solutions of the full system near trajectories of the Kasner map that pass arbitrarily close 
to the Taub points. 
However, their analysis seems to be restricted to the vacuum case.


\section{The Bianchi model}
\label{secBianchiSystem}

We consider the Bianchi class-A model 
on $(N_1,N_2,N_3,\Sigma_+,\Sigma_-)\in\setR^5$ 
with time-derivative $'=d/d\tau$.
In terms of the spatial curvature variables $N_i$ and the shear variables $\Sigma_\pm$ it reads
\begin{equation}\label{eq5dimBianchi}
\begin{array}{rcl}
   N_1'      &=& ( q - 4 \Sigma_+                     ) N_1,
\\ N_2'      &=& ( q + 2 \Sigma_+ + 2\sqrt{3}\Sigma_- ) N_2,
\\ N_3'      &=& ( q + 2 \Sigma_+ - 2\sqrt{3}\Sigma_- ) N_3,
\\ \Sigma_+' &=& -(2-q) \Sigma_+ - 3 S_+,
\\ \Sigma_-' &=& -(2-q) \Sigma_- - 3 S_-.
\end{array}
\end{equation}
The abbreviations
\begin{equation}\label{eq5dimSymbols}
\begin{array}{rcl}
   q         &=& 2 \left( \Sigma_+^2 + \Sigma_-^2 \right)
                 + \frac{1}{2}(3\gamma-2)\Omega,
\\ \Omega    &=& 1 - \Sigma_+^2 - \Sigma_-^2
                   - K,
\\ K         &=& \frac{3}{4} \left( \rule[2ex]{0pt}{0pt} N_1^2 + N_2^2 + N_3^2
                                    - 2\left( N_1 N_2 + N_2 N_3 + N_3 N_1 \right) 
                             \right),
\\ S_+       &=& \frac{1}{2}\left( 
                   \left( N_2 - N_3 \right)^2
                   - N_1 \left( 2 N_1 - N_2 - N_3 \right) 
                 \right),
\\ S_-       &=& \frac{1}{2}\sqrt{3} 
                 \left( N_3 - N_2 \right) \left( N_1 - N_2 - N_3 \right).
\end{array}
\end{equation}
include the deceleration parameter $q$, the density parameter $\Omega$, and the curvature parameter $K$.
The fixed parameter $\frac{2}{3} < \gamma \le 2$, given by the equation of state of an ideal fluid, 
describes the uniformly distributed matter. 
For example, a value $\gamma=1$ represents dust, whereas $\gamma=4/3$ represents radiation.

For a derivation of these equations from the Einstein field equations see \cite{WainwrightHsu1989-BianchiA}, \cite{WainwrightEllis2005-Cosmology},
the appendix to \cite{Ringstroem2001-BianchiIX}, or \cite{Rinstroem2009-CauchyGeneralRelativity}.
\footnote{Note that \cite{WainwrightEllis2005-Cosmology} uses a slightly different scaling from \cite{WainwrightHsu1989-BianchiA} and \cite{Ringstroem2001-BianchiIX}, $N_i^{[WE]} = 3N_i^{[WH]}$ with all other variables being the same. We use the scaling of \cite{WainwrightHsu1989-BianchiA}.}

The resulting flow in $\Omega$ yields
\begin{equation}
\begin{array}{rcl}
\Omega' &=& \left( 2q - (3\gamma-2) \right) \Omega.
\end{array}
\end{equation}
The invariant set $\{\Omega=0\}$ corresponds to 
the 4-dimensional vacuum model
\begin{equation}\label{eq4dimBianchiVacuum}
\begin{array}{rcl}
   N_1'      &=& 2( 1 - K - 2 \Sigma_+                  ) N_1,
\\ N_2'      &=& 2( 1 - K + \Sigma_+ + \sqrt{3}\Sigma_- ) N_2,
\\ N_3'      &=& 2( 1 - K + \Sigma_+ - \sqrt{3}\Sigma_- ) N_3,
\\ \Sigma_+' &=& -2 K \Sigma_+ - 3 S_+,
\\ \Sigma_-' &=& -2 K \Sigma_- - 3 S_-, \qquad \mbox{resp.} \quad 
   \Sigma_-  \;=\; \pm\sqrt{1-K-\Sigma_+^2}.
\end{array}
\end{equation}

Symmetries are given by permutations of $\{N_1,N_2,N_3\}$
together with appropriate linear transformation of $\Sigma_+,\Sigma_-$ corresponding to a
representation of $S_3$ on $\setR^2$.
Together with the reflection $(N_1,N_2,N_3)\mapsto(-N_1,-N_2,-N_3)$,
the system yields a $S_3 \times \setZ_2$ symmetry group.

Note the classification of restrictions of the dynamical system 
to the various invariant regions, see table \ref{tabBianchiClasses}.

\begin{table}
\centering
\begin{tabular}{|c|ccc|}
\hline
Bianchi Class & $N_1$ & $N_2$ & $N_3$ \\
\hline
I       & $0$&$0$&$0$ \\
II      & $+$&$0$&$0$ \\
VI$_0$  & $0$&$+$&$-$ \\
VII$_0$ & $0$&$+$&$+$ \\
VIII    & $-$&$+$&$+$ \\
IX      & $+$&$+$&$+$ \\
\hline
\end{tabular}
\caption{\label{tabBianchiClasses}
Bianchi classes given by the signs of the spatial curvature variables $N_i$. 
Remaining cases are related by equivariance.}
\end{table}


The Kasner circle $\mathcal{K} = \{ N_1=N_2=N_3=0,\; \Omega=0 \}$, 
Bianchi class I,  consists of equilibria.
The attached half ellipsoids $\mathcal{H}_k = \{ N_k \ne 0, \; N_l=N_m=0, \; \Omega=0 \}$, $\{k,l,m\}=\{1,2,3\}$, 
Bianchi class II, 
consist of heteroclinic orbits to equilibria on the Kasner circle, 
see figure \ref{figKasnerCaps}. 
The projections of the trajectories of Bianchi class-II vacuum solutions onto the $\Sigma_\pm$-plane 
yield straight lines through the point $(\Sigma_+,\Sigma_-) = (2,0)$ in the cap $\{N_1 \ne 0, \; N_2=N_3=0\}$.
The projections of the other caps are given by the equivariance.

\begin{figure}
\centering
\setlength{\unitlength}{0.5\textwidth}
\begin{picture}(0.9,0.7)(-0.45,-0.15)
\put(-0.5,0.65){\makebox(0,0)[tl]{\includegraphics[width=\unitlength]{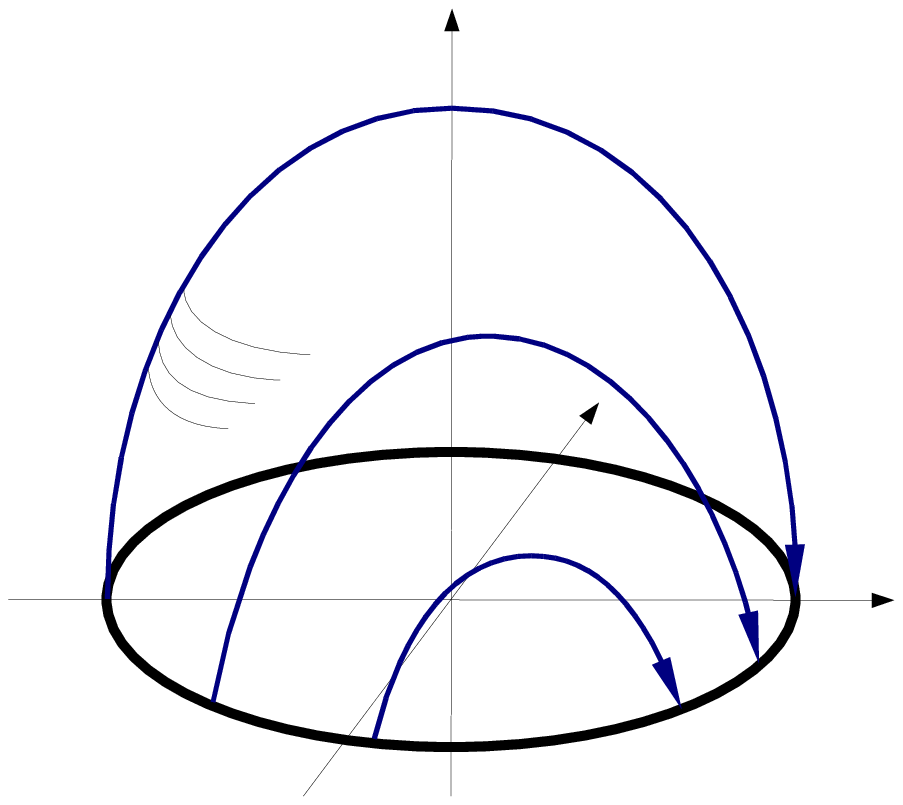}}}
\put(0.47,-0.02){\makebox(0,0)[tr]{$\Sigma_+$}}
\put(0.15,0.17){\makebox(0,0)[br]{$\Sigma_-$}}
\put(0.01,0.6){\makebox(0,0)[tl]{$N_1$}}
\put(0.2,-0.15){\makebox(0,0)[tl]{$\mathcal{K}$}}
\put(-0.23,0.27){\makebox(0,0)[bl]{$\mathcal{H}_1^+$}}
\end{picture}
\hspace*{\fill}
\setlength{\unitlength}{0.01125\textwidth}
\begin{picture}(40,35)(-15,-17.5)
\linethickness{0.4pt}
\put(-15,  0){\vector(1,0){40}}\put(26,-1){\makebox(0,0)[rt]{$\Sigma_+$}}
\put(  0,-20){\vector(0,1){40}}\put(0, 20){\makebox(0,0)[rt]{$\Sigma_-$}}
\put(0,0){\linethickness{1.5pt}\circle{20}}
\put(-10, 17.32){\line(  0,  -1){34.64}}
\put(-10, 17.32){\line(750,-433){30}}
\put(-10,-17.32){\line(750, 433){30}}
\linethickness{0.1pt}
\put(-10, 13.84){\line(750,-346){30}}
\put(-10, 10.40){\line(750,-260){30}}
\put(-10,  6.92){\line(750,-173){30}}
\put(-10,  3.48){\line(750, -87){30}}
\put(-10, -3.48){\line(750,  87){30}}
\put(-10, -6.92){\line(750, 173){30}}
\put(-10,-10.40){\line(750, 260){30}}
\put(-10,-13.84){\line(750, 346){30}}
\linethickness{0.4pt}
\put(-10  ,-17.5){\line(0,-1){ 1.0}}\put(-10  ,-20){\makebox(0,0)[b]{$-1$}}
\put( 20  , -0.2){\line(0,-1){18.3}}\put( 20  ,-20){\makebox(0,0)[b]{$+2$}}
\put( -9.8, 17.32){\line(1, 0){ 9.6}}\put(0, 17.32){\makebox(0,0)[l]{$\sqrt{3}$}}
\put(6.5,-8.5){\makebox(0,0)[tl]{$\mathcal{K}$}}
\put(-2,2){\makebox(0,0)[br]{$\mathcal{H}_1$}}
\end{picture}
\caption{\label{figKasnerCaps}
  Heteroclinic caps of vacuum Bianchi II solutions to the Kasner circle.}
\end{figure}

Away from the singular points, $T_k$, $k=1,2,3$, 
the Kasner circle $\mathcal{K}$ is normally hyperbolic 
with 2-dimensional center-stable manifold 
given by the family of incoming heteroclinic orbits.

The Kasner map $\Phi:\mathcal{K}\to\mathcal{K}$ is defined as follows: 
for each point $q_+\in\mathcal{K}\setminus\{T_1,T_2,T_3\}$ there exists 
a Bianchi class-II vacuum heteroclinic orbit $q(t)$ converging to $q_+$ as $t\to\infty$. 
This orbit is unique up to reflection $(N_1,N_2,N_3)\mapsto(-N_1,-N_2,-N_3)$. 
Its unique $\alpha$-limit $q_-$ defines the image of $q_+$ under the Kasner map
\begin{equation}\label{eqKasnerMap}
  \Phi(q_+) := q_-
\end{equation}
Including the three fixed points, $\Phi(T_k):=T_k$, this construction yields a continuous map,
$\Phi:\mathcal{K}\to\mathcal{K}$. 
In fact $\Phi$ is a non-uniformly expanding map and its image $\Phi(\mathcal{K})$ is a double cover of $\mathcal{K}$.

The main goal are rigorous results on the correspondence of iterations of the Kasner map $\Phi$ to the
dynamics of nearby trajectories to the Bianchi system (\ref{eq5dimBianchi}) with reversed time, 
i.e.~in the $\alpha$-limit $t\to -\infty$.

There exists a 3-cycle of heteroclinic orbits, i.e.~a fixed point of $\Phi^3$, unique up to equivariance.
In the following we choose the cycle given by a heteroclinic orbit $q(t)$ in the 
Bianchi class-II vacuum cap $\{N_1>0, N_2=N_3=0\}$ with $\Sigma_->0$ 
and its images under the equivariances given by cyclic permutation of $N_i$ 
and rotation by $2\pi/3$ in $(\Sigma_+,\Sigma_-)$, see figure \ref{figKasnerCycle}.
The $\alpha$-limit of $q(t)$ is given by
\begin{equation}\label{eq3CycleKasnerPoint}
  (\Sigma_+,\Sigma_-) = q_- = \frac{1}{8}\left( 1-3\sqrt{5}, \sqrt{3}+\sqrt{15} \right).
\end{equation}
After factoring out the equivariance, the heteroclinic orbit $q(t)$ becomes a homoclinic orbit 
on the orbit space of the equivariance. 
We will therefore discuss the dynamics of nearby trajectories by studying the return map $\Psi$
to a transverse cross section to $h(t)$. This is done in section \ref{secReturnMap}.

\begin{figure}
\centering
\setlength{\unitlength}{0.02\textwidth}
\begin{picture}(40,40)(-15,-20)
\linethickness{0.4pt}
\put(-15,  0){\vector(1,0){40}}\put(25,-0.5){\makebox(0,0)[rt]{$\Sigma_+$}}
\put(  0,-20){\vector(0,1){40}}\put(0.3, 20){\makebox(0,0)[lt]{$\Sigma_-$}}
\put(0,0){\linethickness{2pt}\circle{20}}
\put(-10, 17.32){\line(  0,  -1){34.64}}
\put(-10, 17.32){\line(750,-433){30}}
\put(-10,-17.32){\line(750, 433){30}}
\put(-10, 0   ){\circle*{0.5}}\put(-10.3,-0.5){\makebox(0,0)[rt]{$T_1$}}
\put(  5, 8.66){\circle*{0.5}}\put(  5.3, 8.8){\makebox(0,0)[lb]{$T_2$}}
\put(  5,-8.66){\circle*{0.5}}\put(  5.3,-9  ){\makebox(0,0)[lt]{$T_3$}}
\linethickness{0.1pt}
\put(-10, 13.84){\line(750,-346){30}}
\put(-10, 10.40){\line(750,-260){30}}
\put(-10,  6.92){\line(750,-173){30}}
\put(-10,  3.48){\line(750, -87){30}}
\put(-10, -3.48){\line(750,  87){30}}
\put(-10, -6.92){\line(750, 173){30}}
\put(-10,-10.40){\line(750, 260){30}}
\put(-10,-13.84){\line(750, 346){30}}
\put(-10,-17.32){\line(675, 476){27}}
\put(-10,-17.32){\line(600, 520){24}}
\put(-10,-17.32){\line(525, 563){21}}
\put(-10,-17.32){\line(450, 606){18}}
\put(-10,-17.32){\line(375, 649){15}}
\put(-10,-17.32){\line(300, 693){12}}
\put(-10,-17.32){\line(225, 736){ 9}}
\put(-10,-17.32){\line(150, 779){ 6}}
\put(-10,-17.32){\line( 75, 823){ 3}}
\put(-10, 17.32){\line(675,-476){27}}
\put(-10, 17.32){\line(600,-520){24}}
\put(-10, 17.32){\line(525,-563){21}}
\put(-10, 17.32){\line(450,-606){18}}
\put(-10, 17.32){\line(375,-649){15}}
\put(-10, 17.32){\line(300,-693){12}}
\put(-10, 17.32){\line(225,-736){ 9}}
\put(-10, 17.32){\line(150,-779){ 6}}
\put(-10, 17.32){\line( 75,-823){ 3}}
\linethickness{0.4pt}
\put( -2.5,-10  ){\line(0,-1){ 8.5}}\put( -2.5,-20){\makebox(0,0)[b]{
  $-\frac{1}{4}$}}
\put(-10  ,-17.5){\line(0,-1){ 1.0}}\put(-10  ,-20){\makebox(0,0)[b]{$-1$}}
\put( 20  , -0.2){\line(0,-1){18.3}}\put( 20  ,-20){\makebox(0,0)[b]{$+2$}}
\put( -9.8, 17.32){\line(1, 0){ 9.6}}\put(  0.2, 17.32){\makebox(0,0)[l]{
  $\sqrt{3}$}}
%
%
\linethickness{1pt}
\put(-7.13525, 7.00629){\line( 984,-254){16.77050}}
\put( 9.63525, 2.67617){\line(-978,-996){12.13525}}
\put(-2.5    ,-9.68246){\line(-275, 990){ 4.63525}}
\put(-7.13525, 7.00629){\vector( 984,-254){8.38525}}
\put( 9.63525, 2.67617){\vector(-978,-996){6.06762}}
\put(-2.5    ,-9.68246){\vector(-275, 990){2.31762}}
\put(1,5){\makebox(0,0)[bl]{\boldmath$q(t)$}}
\put(-7,7){\makebox(0,0)[br]{\boldmath$q_-$}}
\put(10,2.5){\makebox(0,0)[bl]{\boldmath$q_+$}}
\end{picture}
\caption{\label{figKasnerCycle}
  Kasner circle with sketch of the phase portrait 
  of all Bianchi-II-ellipsoids 
  and the resulting heteroclinic 3-cycle.}
\end{figure}
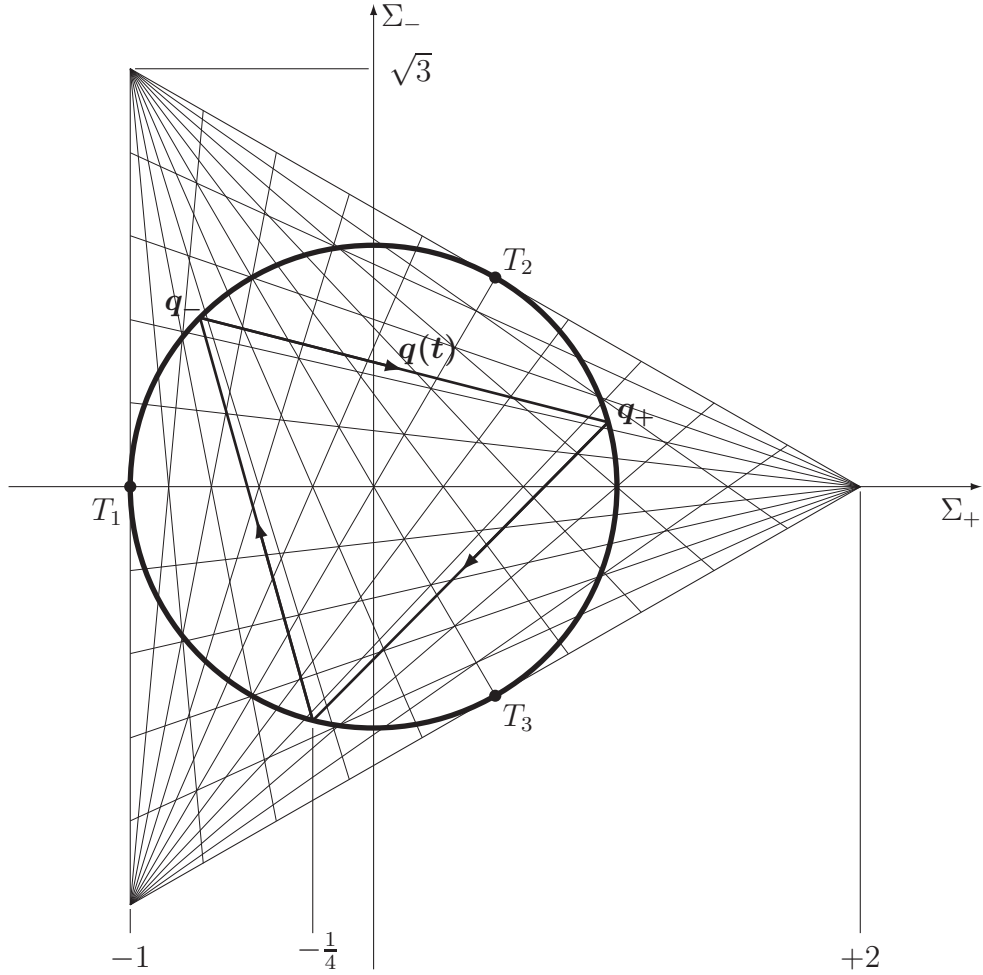

The linearization of (\ref{eq5dimBianchi}) at the Kasner circle yields
\begin{equation}\label{eqKasnerLinearization}
\left(\begin{array}{ccccc}
2 - 4 \Sigma_+ \hspace{-1em}& 0 & 0 & 0 & 0 \\
0 & 2 + 2 \Sigma_+ + 2\sqrt{3} \Sigma_- \hspace{-2em}& 0 & 0 & 0 \\
0 & 0 & 2 + 2 \Sigma_+ - 2\sqrt{3} \Sigma_- \hspace{-2em}& 0 & 0 \\
0 & 0 & 0 & 3(2-\gamma)\Sigma_+^2 & 3(2-\gamma)\Sigma_+\Sigma_- \\
0 & 0 & 0 & 3(2-\gamma)\Sigma_+\Sigma_- & 3(2-\gamma)\Sigma_-^2
\end{array}\right).
\end{equation}
We find eigenvalues
\[
\begin{array}{rcl}
\mu_1 &=& 2-4\Sigma_+,\\
\mu_2 &=& 2 + 2 \Sigma_+ + 2\sqrt{3}\Sigma_-, \\
\mu_3 &=& 2 + 2 \Sigma_+ - 2\sqrt{3}\Sigma_-
\end{array}
\]
to eigenvectors $\partial_{N_1}$, $\partial_{N_2}$, $\partial_{N_3}$ tangential 
to the Bianchi class-II vacuum heteroclinics.
Additionally, there is the trivial eigenvalue zero to the eigenvector 
$-\Sigma_-\partial_{\Sigma_+} + \Sigma_+\partial_{\Sigma_-}$ 
tangential to the Kasner circle $\mathcal{K}$.
The fifth eigenvalue $\mu_\Omega = 3(2-\gamma) > 0$ corresponds to the eigenvector 
$\Sigma_+\partial_{\Sigma_+} + \Sigma_-\partial_{\Sigma_-}$ transverse to the vacuum boundary 
$\{\Omega=0\}$.

At the Kasner equilibrium $q_-$ of the 3-cycle (\ref{eq3CycleKasnerPoint}) we have
\[
\mu_1 = \frac{3}{2}(1+\sqrt{5}), \qquad \mu_2 = 3, \qquad \mu_3 = \frac{3}{2}(1-\sqrt{5})
\]
Note that both unstable eigenvalues are stronger than the stable one, $0 < -\mu_3 < \mu_2 < \mu_1$, 
and that the heteroclinic orbit belonging to the 3-cycle is tangent to the strong unstable direction $\partial_{N_1}$.

In fact at every point on $\mathcal{K}\setminus\{T_1,T_2,T_3\}$ there is one negative eigenvalue 
and it is weaker than the other two positive eigenvalues among $\mu_1, \mu_2, \mu_3$. 
At the singular Taub points $\{T_1,T_2,T_3\}$ two eigenvalues change their signs simultaneously.

We therefore study the local passage near such nonsingular points in general systems in the following section, 
before combining it with the global excursion map 
and applying the results to the particular Bianchi system in section \ref{secReturnMap}.


\section{Local map}
\label{secLocalMap}

In this section we study the passage of trajectories under 
a general flow near a line of equilibria with eigenvalue constraints (\ref{eqGeneralEigenvalueRelation}) 
consistent with the Kasner circle in the Bianchi system.
We will collect estimates on expansion and contraction rates
to establish Lipschitz properties of the local map between sections to a reference orbit 
given by the passage near the line of equilibria, see theorem \ref{thLipschitzLocalMap} at the end of this section.

Consider a $\mathcal{C}^k$ vector field, $k\ge 4$,
\begin{equation}\label{eqGeneralLocalVectorField}
x' = f(x), \qquad x\in\setR^4,
\end{equation}
near a point (w.l.o.g. the origin) on an equilibrium line
\begin{equation}\label{eqGeneralEqLine}
f \left((0,0,0,x_c)^\mathrm{T}\right) \equiv 0, \qquad \forall x_c \in \setR.
\end{equation}
Assume that the linearization at the origin (and then by continuity locally all along the line) 
has the form
\begin{equation}\label{eqGeneralCrudeLinearization}
A(x_c) = Df \left((0,0,0,x_c)^\mathrm{T}\right) = 
\left(\begin{array}{cccc}
\mu_u(x_c) & 0 & 0 & 0 \\
0 & -\mu_s(x_c) & 0 & 0 \\
0 & 0 & -\mu_{ss}(x_c) & 0 \\
\ast & \ast & \ast & 0 
\end{array}\right)
\end{equation}
with
\begin{equation}\label{eqGeneralEigenvalueRelation}
0 < \mu_u < \mu_s < \mu_{ss}.
\end{equation}
We denote $x=(x_u,x_s,x_{ss},x_c)^\mathrm{T}$ according to the above splitting. 
We also abbreviate the vector of stable components as $x_{s,ss}=(x_s,x_{ss})^\mathrm{T}$.

The aim is to study a local map from an in-section $\Sigma^\mathrm{in}=\{x_{ss}=\varepsilon\}$ to an 
out-section $\Sigma^\mathrm{out}=\{x_u=\varepsilon\}$ for $x_c, x_s \approx 0$, see figure \ref{figLocalMap}.
This corresponds to the passage near the Kasner circle in the Bianchi system 
in backwards time direction. (We reversed the time direction to obtain a well defined local map.)

\begin{figure}
\setlength{\unitlength}{0.01125\textwidth}
\centering
\begin{picture}(40,30)(-5,-5)
\thicklines
\put(  0,  0){\line(1,0){12}}\put( 20,  0){\line(1,0){15}}\put(35,-0.5){\makebox(0,0)[rt]{$x_{ss}$}}
\put(  0,  0){\line(0,1){16}}\put(  0, 20){\line(0,1){5}}\put(-0.5, 25){\makebox(0,0)[rt]{$x_u$}}
\put(-5,-2.5){\line(2,1){40}}\put(35,17.5){\makebox(0,0)[rb]{$x_s$}}
\put( -5, -5){\thinlines\line(1,1){30}}\put(25,25){\makebox(0,0)[lt]{$x_c$}}
\put(  5,  0){\vector(-1,0){1}}
\put(  6,  0){\vector(-1,0){1}}
\put(  7,  0){\vector(-1,0){1}}
\put(  5,  2.5){\vector(-2,-1){1}}
\put(  6,  3){\vector(-2,-1){1}}
\put( -5, -2.5){\vector(2,1){1}}
\put( -6, -3){\vector(2,1){1}}
\put(  0,  6){\vector(0,1){1}}
\thinlines
\put(20,0){%
\put( -8, -4){\line(2,1){16}}
\put(  8,  4){\line(0,1){8}}
\put( -8, -4){\line(0,1){8}}
\put( -8,  4){\line(2,1){16}}
\put(  8,  8){\makebox(0,0)[l]{$\,\Sigma^\mathrm{in}$}}
}
\thinlines
\put(0,20){%
\put( -8, -4){\line(2,1){16}}
\put(  8,  4){\line(1,0){10}}
\put( -8, -4){\line(1,0){10}}
\put(  2, -4){\line(2,1){16}}
\put( 13,  4.5){\makebox(0,0)[b]{$\Sigma^\mathrm{out}$}}
}
\thicklines
\put(25,5){\line(1,0){10}}
\put(4,21){\line(0,1){4}}
\qbezier(14,5)(4,5)(4,17)
\put(5,10.7){\vector(-2,5){0.1}}
\put(25,5){\circle*{0.5}}
\put(4,21){\circle*{0.5}}
\put(25,5){\makebox(0,0)[br]{$x^\mathrm{in}$}}
\put(4,21){\makebox(0,0)[tl]{$x^\mathrm{out}$}}
\end{picture}
\caption{\label{figLocalMap}Local passage $\Psi^\mathrm{loc}:\Sigma^\mathrm{in}\to\Sigma^\mathrm{out}$.}
\end{figure}

Rescaling coordinates yields a system
\begin{equation}\label{eqGeneralLocalVectorFieldLinNonlin}
x' = A(x_c)x + \varepsilon g(x),
\end{equation}
with $\varepsilon$ arbitrarily fixed and $g$ quadratic in $(x_u, x_s, x_{ss})$.
The local map 
\begin{equation}\label{eqLocalMap}
  (x_u^\mathrm{in}, x_s^\mathrm{in}, x_c^\mathrm{in}) \longmapsto 
  (x_s^\mathrm{out}, x_{ss}^\mathrm{out}, x_c^\mathrm{out}) = 
    \Psi^\mathrm{loc} (x_u^\mathrm{in}, x_s^\mathrm{in}, x_c^\mathrm{in})
\end{equation}
is given by the first intersection of the solution of (\ref{eqGeneralLocalVectorFieldLinNonlin}) 
to the initial value $(x_u^\mathrm{in}, x_s^\mathrm{in}, x_{ss}^\mathrm{in}=1, x_c^\mathrm{in})$
with the out-section $\{x_u=1\}$. 
The singular points $x^\mathrm{in}$ in the intersection of the stable manifold of the equilibrium line with the in-section
are mapped to the respective points in the intersection of the unstable manifold of the equilibrium line with the out-section.
Time at the in-section is usually set to zero. 
Time at the out-section depends on the initial condition and is denoted by $t^\mathrm{loc}$.

For any $x_c$ fixed and any $k\in\setN$, the equilibrium $(0,0,0,x_c)^\mathrm{T}$ possesses a
one-dimensional unstable manifold $W^u(x_c)$ and a stable manifold $W^{s,ss}(x_c)$. 

Since the stable and unstable manifolds as well as the strong stable foliation of the stable manifold
are $\mathcal{C}^k$ they can be flattened, see e.g. \cite{ShilnikovTuraevChua1998-MethodsNonlinearDynamicsI}, Theorem 5.8.
By a $\mathcal{C}^k$ change of coordinate the stable / strong stable / unstable manifolds to the 
equilibria locally coincide with the respective eigenspaces, in particular the following subspaces become invariant:
\begin{equation}\label{eqGeneralInvManifolds}
\begin{array}{rcl}
W^u(x_c) &=& \{x_s = x_{ss} = 0, x_c \mbox{ fixed} \}, \\
W^{s,ss}(x_c) &=& \{x_u = 0, x_c \mbox{ fixed} \}, \\
W^{ss}(x_c) &=& \{x_u = x_s = 0, x_c \mbox{ fixed} \}.
\end{array}
\end{equation}
In fact, in the Bianchi system, $W^u(x_c)$ and $W^{ss}(x_c)$ coincide with the class-II caps formed by families of incoming and outgoing 
heteroclinic orbits.

The local map $\Psi^\mathrm{loc}$ is well-defined on the in-section
\begin{equation}\label{eqInSection}
\Sigma^\mathrm{in} = 
  \{\,(x_u^\mathrm{in}, x_s^\mathrm{in}, x_{ss}^\mathrm{in}, x_c^\mathrm{in})\,|\,
  x_{ss}^\mathrm{in} = 1, 0 < x_u^\mathrm{in} < 1, |x_s^\mathrm{in}| < 1, |x_c^\mathrm{in}| < 1\,\}
\end{equation}
see lemma \ref{thLocalPassage} below.
The singular points of the local map thus become the set $\{x_u^\mathrm{in}=0\}$ and we define:
\begin{equation}\label{eqLocalMapExtension}
\Psi^\mathrm{loc}(x_u^\mathrm{in}=0, x_s^\mathrm{in}, x_c^\mathrm{in}) \;=\; 
  (x_s^\mathrm{out}, x_{ss}^\mathrm{out}, x_c^\mathrm{out}) \,:=\;
  (0,0,x_c^\mathrm{in}).
\end{equation}

Let $P_u$, $P_s$, $P_{ss}$, $P_{s,ss}:= P_s+P_{ss}$, $P_c$ be the eigenprojections with respect to $A(x_c)$. 
Then (\ref{eqGeneralCrudeLinearization}) yields
\begin{equation}\label{eqGeneralLinearization}
A(x_c) = Df \left((0,0,0,x_c)^\mathrm{T}\right) = 
\left(\begin{array}{cccc}
\mu_u(x_c) & 0 & 0 & 0 \\
0 & -\mu_s(x_c) & 0 & 0 \\
0 & 0 & -\mu_{ss}(x_c) & 0 \\
0 & 0 & 0 & 0 
\end{array}\right)
\end{equation}
and due to (\ref{eqGeneralInvManifolds}) the higher order terms have the form
\begin{equation}\label{eqGeneralHot}
\begin{array}{rclcll}
P_u g(x) &=& x_u \tilde{g}_u(x) &&&\in\setR,\\
P_{s,ss} g(x) &=& \tilde{g}_{s,ss}(x) x_{s,ss} 
  &=& \left(\begin{array}{cc} \tilde{g}_{s1}(x) & \tilde{g}_{s2}(x) \\ \tilde{g}_{ss1}(x) & \tilde{g}_{ss2}(x) \end{array}\right) 
      \left({ x_s \atop x_{ss} }\right) 
  &\in\setR^2, \\
P_c g(x) &=& x_u \tilde{g}_{c}(x) x_{s,ss} 
  &=& x_u \, \left( \begin{array}{cc} \tilde{g}_{c1}(x) & \tilde{g}_{c2}(x) \end{array} \right) 
      \left({ x_s \atop x_{ss} }\right) 
  &\in\setR,
\end{array}
\end{equation}
with $\mathcal{C}^{k-1}$-functions $\tilde{g}_u$, $\tilde{g}_{s,ss}$, vanishing along the line of equilibria, and 
$\mathcal{C}^{k-2}$-function $\tilde{g}_{c}$.
In particular
\begin{equation}\label{eqGeneralHotBounds}
\begin{array}{rcl}
|\tilde{g}_{c1}(x)|, |\tilde{g}_{c2}(x)|& < & C, \\|\tilde{g}_{u}(x)|, |\tilde{g}_{s1}(x)|, |\tilde{g}_{s2}(x)|, |\tilde{g}_{ss1}(x)|, |\tilde{g}_{ss2}(x)| & < & C \max(|x_u|,|x_s|,|x_{ss}|)
\end{array}
\end{equation}
for some constant $C>0$ independent of $\varepsilon$ and $x\in\mathcal{U}$, 
where $\mathcal{U}$ is some local neighborhood of the origin. 
Similarly, $\tilde{g}$ satisfies Lipschitz bounds
\begin{equation}\label{eqGeneralHotLipschitzBounds}
|\tilde{g}_c(x)-\tilde{g}_c(\tilde{x})|, \quad
|\tilde{g}_u(x)-\tilde{g}_u(\tilde{x})|, \quad
\|\tilde{g}_{s,ss}(x)-\tilde{g}_{s,ss}(\tilde{x})\| < C \|x-\tilde{x}\|.
\end{equation}
%
Norms of vectors are always taken as $\ell_1$-norms, e.g. 
\[
\|x\|=|x_c|+|x_u|+\|x_{s,ss}\| = |x_c|+|x_u|+|x_s|+|x_{ss}|.
\]
We choose
\begin{equation}\label{eqLocalNeighborhood}
\mathcal{U}=(-2,2)^4.
\end{equation}
All further estimates will use this rescaled system (\ref{eqGeneralLocalVectorFieldLinNonlin}) 
with flattened invariant manifolds (\ref{eqGeneralInvManifolds}) in the local neighborhood $\mathcal{U}$.
They will be valid for all $\varepsilon<\varepsilon_0$ and suitably chosen $\varepsilon_0$. 
In the original system (\ref{eqGeneralLocalVectorField}) $\varepsilon_0$ bounds the size of the neighborhood 
of the origin in which this local analysis is valid.

\begin{rem} The invariance of $W^{ss}(x_c)$ implies that $\tilde{g}_{s2}(x_u=0,x_s,x_{ss},x_c)=0$. We do however
not exploit this fact in our analysis below.
\end{rem}

\begin{prop}\label{thEigenvalueBounds}
Let
\[
\mu_u := \mu_u(0), \quad -\mu_s := -\mu_s(0), \quad -\mu_{ss} := -\mu_{ss}(0)
\]
be the eigenvalues of (\ref{eqGeneralLinearization}) at the origin.
Then for all $0<\alpha<1$ there exists an $\varepsilon_0>0$ such that 
for all $\varepsilon<\varepsilon_0$ in (\ref{eqGeneralLocalVectorFieldLinNonlin}) and $x\in\mathcal{U}$
\[
\alpha \;\le\; \frac{\mu_u(x_c)}{\mu_u}, \frac{\mu_s(x_c)}{\mu_s}, \frac{\mu_{ss}(x_c)}{\mu_{ss}} \;\le\; \alpha^{-1}.
\]
\end{prop}
\begin{proof}
The distinct eigenvalues in the linearization of the original system (\ref{eqGeneralCrudeLinearization})
depend differentiably on $x_c$, as long as (\ref{eqGeneralEigenvalueRelation}) holds. 
For the rescaled system (\ref{eqGeneralLocalVectorFieldLinNonlin}) with small $\varepsilon_0$ this provides 
bounds in $\mathcal{U}$: 
Indeed, there exists a constant $C>0$ independent of $\varepsilon_0$, $\varepsilon$, such that
\begin{equation}\label{eqEigenvalueGradients}
\textstyle
\left| \frac{\diff}{\diff x_c}\mu_{u}(x_c) \right|,\;
\left| \frac{\diff}{\diff x_c}\mu_{s}(x_c) \right|,\;
\left| \frac{\diff}{\diff x_c}\mu_{ss}(x_c) \right| \;<\; \varepsilon C.
\end{equation}
\end{proof}

The scalar function $\theta(x):= \mu_u \left(\mu_u(x_c) + \varepsilon\tilde{g}_u(x)\right)^{-1}$ is therefore $\mathcal{C}^{k-1}$ and close to $1$.
The vector field 
\[
  x' \;=\; \theta(x)f(x) \;=\; \frac{\mu_u}{\mu_u(x_c) + \varepsilon\tilde{g}_u(x)} f(x)
\]
has the same trajectories as the original vector field and all previous considerations remain valid. 
Thus we can assume, without loss of generality, that $\theta(x) \equiv 1$ in $\mathcal{U}$, i.e.
\begin{equation}\label{eqEulerMultiplier}
 \mu_u(x_c) \;\equiv\; \mu_u, \qquad \tilde{g}_u(x) \equiv 0.
\end{equation}
At this step we have made use of the fact that the origin possesses exactly one unstable eigenvalue. 
The vector field to consider then has the form
\begin{equation}\label{eqSplitLocalSystem}
\begin{array}{rcl}
x_u' &=& \mu_u \, x_u,
\\
x_{s,ss}' &=& - \mu_{s,ss}(x_c) \, x_{s,ss} + \varepsilon \, \tilde{g}_{s,ss}(x) \, x_{s,ss},
\\
x_c' &=& \varepsilon \, x_u \, \tilde{g}_{c}(x) \, x_{s,ss}.
\end{array}
\end{equation}
Note again the abbreviations
\[\begin{array}{rclrcl}
x_{s,ss} &=& \left(\begin{array}{c} x_s \\ x_{ss} \end{array}\right), 
& \qquad
\mu_{s,ss}(x_c) &=& \left(\begin{array}{cc} \mu_s(x_c) & 0 \\ 0 & \mu_{ss}(x_c) \end{array}\right), 
\\
\tilde{g}_{s,ss}(x) &=& \left(\begin{array}{cc} \tilde{g}_{s1}(x) & \tilde{g}_{s2}(x) \\ \tilde{g}_{ss1}(x) & \tilde{g}_{ss2}(x) \end{array}\right),
& \quad
\tilde{g}_c(x) &=& \left( \begin{array}{cc} \tilde{g}_{c1}(x) & \tilde{g}_{c2}(x) \end{array} \right),
\end{array}\]
to facilitate later generalizations to higher dimensions of the stable component $x_{s,ss}$. 
In particular note the diagonal form of the linear part $\mu_{s,ss}$.

\begin{lem}\label{thMonotoneDistance}
Define 
\[
\chi \,:=\; |x_u|\,\|x_{s,ss}\| \,:=\; |x_u|\,|x_s| + |x_u|\,|x_{ss}|
\]
Then for all $0<\alpha<1$ there exists an $\varepsilon_0>0$ such that 
for all $\varepsilon<\varepsilon_0$ and $x\in\mathcal{U}$
\[
\chi' \;\le\; -\alpha(\mu_s-\mu_u) \chi
\]
under the dynamics (\ref{eqSplitLocalSystem}).
In particular, with $x(0) = x^\mathrm{in}$ on the in-section $\Sigma^\mathrm{in}$
as long as the trajectory remains inside $\mathcal{U}$
\[
\chi(t) \;\le\; \exp(-\alpha(\mu_s-\mu_u) t)\,|x_u^\mathrm{in}|\,\|x_{s,ss}^\mathrm{in}\| \;\le\; 2\exp(-\alpha(\mu_s-\mu_u) t)\,|x_u^\mathrm{in}|.
\]
\end{lem}
\begin{rem}
The quantity $\chi$ describes the ``distance'' from the critical heteroclinic cycle. 
More precisely, $\chi$ is related to the distance from the set $\bigcup_{x_c} \left(W^{s,ss}(x_c)\cup W^u(x_c)\right)$.
Close to the critical heteroclinic cycle this set is however under the return map
squeezed onto the critical cycle, so for large times $\chi$ measures in fact a ``distance''
to the critical heteroclinic cycle.
\end{rem}
\begin{proof}
Vectorfield (\ref{eqSplitLocalSystem}) yields
\[
\chi' \;\le\;
-(\mu_s(x_c)-\mu_u)\chi + 2 \varepsilon \|\tilde{g}_{s,ss}(x)\|\chi.
\]
The function $\tilde{g}|_\mathcal{U}$ is bounded uniformly in $\varepsilon$, see (\ref{eqGeneralHotBounds}). 
Additionally, proposition \ref{thEigenvalueBounds} provides bounds $-\mu_s(x_c) \le -\tilde{\alpha}\mu_s$ 
for arbitrary $0 < \tilde{\alpha} < 1$ and $\varepsilon_0$ small enough. 
By a suitable choice of $\tilde{\alpha}$ between $\alpha$ and 1 we estimate
\[
\chi' \;\le\; -\alpha(\mu_s-\mu_u) \chi.
\]
Integrating this differential inequality and using the initial value 
yield the second claim.
\end{proof}

We need estimates for the passage time $t^\mathrm{loc}$ as well as bounds on 
$|x_c|$, $|x_u|$, and $\|x_{s,ss}\|$. 
These are obtained in the following lemmata and summarized in corollary \ref{thContinuousLocalMap}.

\begin{lem}\label{thLocalBounds}
For all $0<\alpha<1$ there exists an $\varepsilon_0>0$ such that 
for all $\varepsilon<\varepsilon_0$ and $x(0)=x^\mathrm{in}$ in the in-section $\Sigma^\mathrm{in}$, see (\ref{eqInSection}),
as long as $x(t)$ remains in $\mathcal{U}$ under the flow to the vector field
(\ref{eqSplitLocalSystem})
\begin{eqnarray}\label{eqLocalUnstableBound}
x_u(t) &=& \exp(\mu_u t) \, x_u^\mathrm{in},
\\\label{eqLocalStableBound} \hspace{-2em}
\|x_{s,ss}(t)\| \;=\; |x_s(t)| + |x_{ss}(t)| &\le& \exp(-\alpha \mu_s t) \, \|x_{s,ss}^\mathrm{in}\|,
\\\label{eqLocalCenterBound}
|x_c(t) - x_c^\mathrm{in}| &\le& \frac{2 \varepsilon C}{\alpha(\mu_s-\mu_u)} |x_u^\mathrm{in}|.
\end{eqnarray}
Here, $C$ is the uniform (in $x$ and $\varepsilon$) bound on $\tilde{g}|_\mathcal{U}$ from (\ref{eqGeneralHotBounds}).
\end{lem}
\begin{proof}
The unstable component (\ref{eqLocalUnstableBound}) is given directly by the vectorfield.

On the stable component, we obtain the estimate
\[
\textstyle \frac{\diff}{\diff t} \left( |x_s| + |x_{ss}| \right) \;\le\;
  \left( -\mu_s(x_c) + 2 \varepsilon \|\tilde{g}_{s,ss}(x)\| \right)
  \,
  \left( |x_s| + |x_{ss}| \right).
\]
The function $\tilde{g}|_\mathcal{U}$ is bounded uniformly in $\varepsilon$, see (\ref{eqGeneralHotBounds}).
Together with the estimate on the eigenvalues of proposition \ref{thEigenvalueBounds}, 
this yields the bounds (\ref{eqLocalStableBound}) on the stable component.

The center component of (\ref{eqSplitLocalSystem}) is estimated by 
\[
\begin{array}{rclcl} 
\displaystyle |x_c(t) - x_c^\mathrm{in}| 
   &=&   \displaystyle \varepsilon \left|\int_0^t P_cg(x(s)) \diff s\right|
   &\le& \displaystyle \varepsilon \int_0^t |P_cg(x(s))| \diff s
\\ &\le& \displaystyle \varepsilon C \int_0^t |\chi(s)| \diff s
   &\le& \displaystyle \varepsilon C \frac{2}{\alpha(\mu_s-\mu_u)} |x_u^\mathrm{in}|.
\end{array}
\]
Here we used the estimate of lemma \ref{thMonotoneDistance} on $\chi$.
\end{proof}

\begin{cor}\label{thLocalPassage}
There exists an $\varepsilon_0>0$ such that 
for all $\varepsilon<\varepsilon_0$ and $x(0)=x^\mathrm{in}$ in the in-section $\Sigma^\mathrm{in}$, see (\ref{eqInSection}),
the trajectory $x(t)$ under the flow to the vector field (\ref{eqSplitLocalSystem})
remains in $\mathcal{U}$ as long as $|x_u| \le 1$, 
i.e.~all along the passage defining the local map $\Psi^\mathrm{loc}$. 
The passage time $t^\mathrm{loc}$ is given by
\[
 t^\mathrm{loc} \;=\; \frac{1}{\mu_u} \ln \frac{1}{|x_u^\mathrm{in}|}.
\]
\end{cor}
\begin{proof}
We choose $\varepsilon_0$ smaller than $\frac{1}{2C}\alpha(\mu_s-\mu_u)$, 
see lemma \ref{thLocalBounds}.
Then trajectories starting in $\Sigma^\mathrm{in}$ cannot leave $\mathcal{U}$ unless $x_u$ becomes larger than 1,
see (\ref{eqLocalCenterBound}), (\ref{eqLocalStableBound}).
Furthermore, (\ref{eqLocalUnstableBound}) ensures that $x_u$ must grow beyond 1.
Thus every trajectory starting in $\Sigma^\mathrm{in}$ intersects the out-section $\{x_u=1\}$ 
before leaving $\mathcal{U}$.
Setting $x_u(t^\mathrm{loc}) = 1$ in (\ref{eqLocalUnstableBound}) determines the passage time $t^\mathrm{loc}$.
\end{proof}

\begin{cor}\label{thContinuousLocalMap}
The local map $\Psi^\mathrm{loc}$ (\ref{eqLocalMap}, \ref{eqLocalMapExtension}), 
i.e.~the local passage on the closed in-section $\overline\Sigma^\mathrm{in}$ including the singular line $\{x_u^\mathrm{in} = 0\}$, 
is continuous.
For all $0<\alpha<1$ there exists an $\varepsilon_0>0$ such that 
for all $\varepsilon<\varepsilon_0$ the following estimates hold
\[\begin{array}{rcl}
  \displaystyle |x_c^\mathrm{out} - x_c^\mathrm{in}| 
    &\le& \displaystyle 
     \frac{2 \varepsilon C}{\alpha(\mu_s-\mu_u)} |x_u^\mathrm{in}|, \\
  \displaystyle \|x_{s,ss}^\mathrm{out}\| \;=\; |x_s^\mathrm{out}| + |x_{ss}^\mathrm{out}| 
    &\le& \displaystyle 
     |x_u^\mathrm{in}|^{\alpha\mu_s/\mu_u - 1} |x_u^\mathrm{in}| \,\|x_{s,ss}^\mathrm{in}\|
    \;\le\;
     |x_u^\mathrm{in}|^{\alpha\mu_s/\mu_u - 1} ( |x_u^\mathrm{in}| + |x_s^\mathrm{in}| ).
\end{array}\] 
Thus the drift along the line of equilibria is arbitrarily small and the distance from the orbit to
the union of the stable and unstable manifolds shrinks arbitrarily fast, close to the critical orbit.
\end{cor}
\begin{proof}
The estimates follow directly from lemma \ref{thLocalBounds}
applied to the local passage time given by corollary \ref{thLocalPassage}.
They also establish continuity of the local map $\Psi^\mathrm{loc}$ at the singular line $\{ x_u^\mathrm{in} = 0 \}$.
Note $0 < x_u^\mathrm{in} < 1$ and $x_{ss}^\mathrm{in} = 1$ on the in-section.
Note further that $\beta = \alpha\mu_s/\mu_u - 1 > 0$ for $\alpha$ chosen close enough to $1$.
\end{proof}

To obtain Lipschitz-bounds for the local map $\Psi^\mathrm{loc}$, we have to improve the estimates considerably. 
Consider two trajectories, $x(t)$ and $\tilde{x}(t)$, starting in the in-section $\Sigma^\mathrm{in}$. 
If $\tilde{x}_u^\mathrm{in}=0$ then $\tilde{x}_{s,ss}^\mathrm{out}=0$, $\tilde{x}_c^\mathrm{out}=\tilde{x}_c^\mathrm{in}$,
and the Lipschitz estimates follow from corollary \ref{thContinuousLocalMap}.
Choose, without loss of generality, 
\[
 0 < \tilde{x}_u^\mathrm{in} \le x_u^\mathrm{in}.
\]
In order to compare the two trajectories, we want $\tilde{x}_u(0) = x_u^\mathrm{in} = x_u(0)$ and therefore 
start the trajectory $\tilde{x}$ at negative time
\begin{equation}\label{eqTimeInitialSegment}
  \tilde{x}(-\tilde{t}^0) \;=\; \tilde{x}^\mathrm{in}, \qquad \tilde{t}^0\,:=\; \frac{1}{\mu_u} \ln \frac{x_u^\mathrm{in}}{\tilde{x}_u^\mathrm{in}}.
\end{equation}
Then $\tilde{x}_u(t) = x_u(t)$ for all $t\ge 0$. In particular, both trajectories meet the out-section at the same
time  $t^\mathrm{loc} = \frac{1}{\mu_u} \ln \frac{1}{|x_u^\mathrm{in}|}$, see corollary \ref{thLocalPassage}.
The next lemma gives some estimates for the initial segment of the trajectory $\tilde{x}$ up to $t=0$.

\begin{lem}\label{thLipEstimateInitialSegment}
Consider (\ref{eqSplitLocalSystem}) with initial values $x(0)=x^\mathrm{in}$ and 
$\tilde{x}(-\tilde{t}^0)=\tilde{x}^\mathrm{in}$ in $\Sigma^\mathrm{in}$. 
Assume without loss of generality $0 < \tilde{x}_u^\mathrm{in} \le x_u^\mathrm{in}$.
There exists an $\varepsilon_0>0$ such that for all $\varepsilon<\varepsilon_0$ the following estimates hold:
\[
\begin{array}{rcl}
| \tilde{x}_c(0) - \tilde{x}_c^\mathrm{in} | &\le& \displaystyle
\varepsilon L_c \, \frac{\tilde{x}_u^\mathrm{in}}{x_u^\mathrm{in}} | x_u^\mathrm{in} - \tilde{x}_u^\mathrm{in} |,
\\[2ex]
\| \tilde{x}_{s,ss}(0) - \tilde{x}_{s,ss}^\mathrm{in} \| \;=\;
| \tilde{x}_s(0) - \tilde{x}_s^\mathrm{in} | + | \tilde{x}_{ss}(0) - \tilde{x}_{ss}^\mathrm{in} | &\le & \displaystyle
  L_s \, \frac{1}{x_u^\mathrm{in}} | x_u^\mathrm{in} - \tilde{x}_u^\mathrm{in} |,
\end{array}
\]
with suitably chosen constants $L_s > 0$, $L_c > 0$, uniform in $x^\mathrm{in}$, $\tilde{x}^\mathrm{in}$ and $\varepsilon$.
\end{lem}
\begin{proof}
%
The $x_c$-component is estimated similar to lemma \ref{thLocalBounds}:
\[
|\tilde{x}_c(t) - \tilde{x}_c^\mathrm{in}| 
   \;\le\; \varepsilon C \int_{-\tilde{t}^0}^0 |\tilde{\chi}(s)| \diff s
   \;\le\; \varepsilon C \int_0^{\tilde{t}^0} \exp(-\alpha(\mu_s-\mu_u) s) |\tilde{x}_u^\mathrm{in}| \,\|\tilde{x}_{s,ss}^\mathrm{in}\|.
\]
Again $C$ is the uniform bound on $\tilde{g}|_\mathcal{U}$. The last inequality uses the estimate of lemma \ref{thMonotoneDistance}.
Plugging in (\ref{eqTimeInitialSegment}) and $\|\tilde{x}_{s,ss}^\mathrm{in}\|\le2$ we obtain
\[
\begin{array}{rcl} 
\displaystyle |x_c(0) - x_c^\mathrm{in}| 
   &\le& \displaystyle \frac{2 \varepsilon C}{\alpha(\mu_s-\mu_u)}|\tilde{x}_u^\mathrm{in}|\left( 1 - \exp(-\alpha(\mu_s-\mu_u)\tilde{t}^0) \right)
\\ &\le& \displaystyle \frac{2 \varepsilon C}{\alpha(\mu_s-\mu_u)}|\tilde{x}_u^\mathrm{in}|\left( 
                1 - \left( \frac{x_u^\mathrm{in}}{\tilde{x}_u^\mathrm{in}} \right)^{-(\mu_s-\mu_u)/\mu_u} \right)
\\ &\le& \displaystyle \frac{2 \varepsilon C L}{\alpha(\mu_s-\mu_u)}|\tilde{x}_u^\mathrm{in}|\left( 
                1 - \frac{\tilde{x}_u^\mathrm{in}}{x_u^\mathrm{in}} \right).
\end{array}
\]
The last inequality with $L := \max\{ (\mu_s-\mu_u)/\mu_u,\, 1 \}$ uses 
the general Lipschitz estimate $1-\xi^\vartheta \le \max\{\vartheta,1\}(1-\xi)$ for arbitrary $0<\xi<1$ and $\vartheta > 0$.
Finally, the choice $L_c:= 2 C L \alpha^{-1} (\mu_s-\mu_u)^{-1}$  yields the claim on the $x_c$-component.

The $x_{s,ss}$-component is contracted at most with the rate of the strong stable eigenvalue $\mu_{ss}$, perturbed by small higher order terms. 
Indeed, due to the diagonal form of the linear part, the variation-of-constant formula reads
\[
\begin{array}{rcl} 
\tilde{x}_{s,ss}(0)
&=& \displaystyle
  \left( \begin{array}{cc} \exp \int_{-\tilde{t}^0}^0 -\mu_s(\tilde{x}_c(s)) \diff s \hspace{-1em} & 0 \\ 
         0 & \hspace{-1em} \exp \int_{-\tilde{t}^0}^0 -\mu_{ss}(\tilde{x}_c(s)) \diff s \end{array} \right)
  \tilde{x}_{s,ss}^\mathrm{in}
\\[3ex] &+& \displaystyle
  \varepsilon \int_{-\tilde{t}^0}^0 
  \left( \begin{array}{cc} \exp \int_s^0 -\mu_s(\tilde{x}_c(\sigma)) \diff \sigma \hspace{-1em} & 0 \\ 
         0 & \hspace{-1em} \exp \int_s^0 -\mu_{ss}(\tilde{x}_c(\sigma)) \diff \sigma \end{array} \right)
  \tilde{g}_{s,ss}(\tilde{x}(s))\tilde{x}_{s,ss} \diff s.
\end{array}
\]
We have already established uniform bounds on the nonlinearity, $|g| < C$, and the contraction rates, 
$\alpha\mu_s < \mu_s(\tilde{x}_c) < \mu_{ss}(\tilde{x}_c) < \frac{1}{\alpha}\mu_{ss}$, 
with $0 < \alpha < 1$ arbitrarily close to $1$. Thus
\[
\begin{array}{rcl}
\| \tilde{x}_{s,ss}(0) - \tilde{x}_{s,ss}^\mathrm{in} \|
  &<& \displaystyle
        \left( 1 - \exp( -{\textstyle\frac{1}{\alpha}}\mu_{ss}\tilde{t}^0 ) \right)
        \| \tilde{x}_{s,ss}^\mathrm{in} \|
  \;+\; \varepsilon C \int_0^{\tilde{t}^0} \exp( -\alpha\mu_s s ) \diff s
\\[1ex] &\le& \displaystyle
        2 \left( 1 - \exp( -{\textstyle\frac{1}{\alpha}}\mu_{ss}\tilde{t}^0 ) \right)
  \;+\; \frac{\varepsilon C}{\alpha\mu_s} \left( 1 - \exp( -\alpha\mu_s \tilde{t}^0 ) \right)
\\[1ex] &=& \displaystyle
        2 \left( 1 - \left( \frac{\tilde{x}_u^\mathrm{in}}{x_u^\mathrm{in}} \right)^{\mu_{ss}/\alpha\mu_u} \right)
  \;+\; \frac{\varepsilon C}{\alpha\mu_s} 
        \left( 1 - \left( \frac{\tilde{x}_u^\mathrm{in}}{x_u^\mathrm{in}} \right)^{\alpha\mu_{s}/\mu_u} \right)
\\[2ex] &\le& \displaystyle
        2 L_1 \left( 1 - \frac{\tilde{x}_u^\mathrm{in}}{x_u^\mathrm{in}} \right)
  \;+\; \frac{2 \varepsilon C L_2}{\alpha\mu_s} 
        \left( 1 - \frac{\tilde{x}_u^\mathrm{in}}{x_u^\mathrm{in}} \right).
\end{array}
\]
The last inequality again uses
the general estimate $1-\xi^\vartheta \le \max\{\vartheta,1\}(1-\xi)$ for arbitrary $0<\xi<1$ and $\vartheta > 0$.
The constants $L_1$ and $L_2$ are independent of $\varepsilon$ and therefore yield the claim.
\end{proof}

\begin{thm}\label{thLipschitzLocalMap}
The local map $\Psi^\mathrm{loc}$ (\ref{eqLocalMap}, \ref{eqLocalMapExtension}), 
i.e.~the local passage on $\overline\Sigma^\mathrm{in}$ including the singular line $x_u^\mathrm{in} = 0$, 
is Lipschitz continuous. 
Let 
\[
\| \tilde{x}^\mathrm{in} - x^\mathrm{in} \|=|\tilde{x}_u^\mathrm{in} - x_u^\mathrm{in}|+|\tilde{x}_c^\mathrm{in} - x_c^\mathrm{in}|+|\tilde{x}_s^\mathrm{in} - x_s^\mathrm{in}|
\]
be the norm in the in-section $\Sigma^\mathrm{in}$.
Then there exists an $\varepsilon_0>0$ such that for all $\varepsilon<\varepsilon_0$ and
$0 \le \tilde{x}_u^\mathrm{in} \le x_u^\mathrm{in}$ the following estimates hold:
\[
\begin{array}{rcl}
| (\tilde{x}_c^\mathrm{out} - x_c^\mathrm{out}) - ( \tilde{x}_c^\mathrm{in} - x_c^\mathrm{in} ) |
  &\le& \varepsilon L_c \| \tilde{x}^\mathrm{in} - x^\mathrm{in} \|,
\\
\| \tilde{x}_{s,ss}^\mathrm{out} - x_{s,ss}^\mathrm{out} \|
  \;=\; | \tilde{x}_{s}^\mathrm{out} - x_{s}^\mathrm{out} | + | \tilde{x}_{ss}^\mathrm{out} - x_{ss}^\mathrm{out} |
  &\le& L_s |x_u^\mathrm{in}|^\beta \| \tilde{x}^\mathrm{in} - x^\mathrm{in} \|,
\end{array}
\]
with $L_c>0$, $L_s>0$, $\beta>0$ uniform in $x$, $\varepsilon$.
In particular, $L_s |x_u^\mathrm{in}|^\beta$ is arbitrarily small for $x_u^\mathrm{in}$ small enough.

The drift in the center direction can be made arbitrarily small by choosing a sufficiently small local neighborhood.
The contraction in the transverse directions is arbitrarily strong by restricting the in-section to the part close to the
primary object, i.e.~the stable manifold of the origin.
\end{thm}
\begin{proof}
If $\tilde{x}_u^\mathrm{in} = 0$ then 
$\tilde{x}_{s,ss}^\mathrm{out} = 0$, $\tilde{x}_{c}^\mathrm{out} = \tilde{x}_{c}^\mathrm{in}$,
and the Lipschitz estimates follow from corollary \ref{thContinuousLocalMap}. 
Therefore let $0 < \tilde{x}_u^\mathrm{in} \le x_u^\mathrm{in}$. 
Choose $\tilde{x}(-\tilde{t}^0) = \tilde{x}^\mathrm{in}$ according to (\ref{eqTimeInitialSegment}). 
Then $\tilde{x}(0)$ is estimated by lemma \ref{thLipEstimateInitialSegment}, and $x(0)=x^\mathrm{in}$.

We abbreviate
\begin{equation}\label{eqDcDsDefinition}
\begin{array}{rclrcl}
 d_c(t) &:=& \displaystyle \int_0^t \left| \left( \tilde{x}_c(s) - x_c(s) \right)' \right| \diff s,
 &\qquad
 d_c(0) &=& 0,
\\
 d_s(t) &:=& \| \tilde{x}_{s,ss}(t) - x_{s,ss}(t) \|,
\end{array}
\end{equation}
In particular, using lemma \ref{thLipEstimateInitialSegment}, we estimate the initial value,
\begin{equation}\label{eqDsInitial}
 d_s(0) \;\le\; L_0 \, \frac{1}{x_u^\mathrm{in}} | x_u^\mathrm{in} - \tilde{x}_u^\mathrm{in} | 
                + \| \tilde{x}_{s,ss}^\mathrm{in} - x_{s,ss}^\mathrm{in} \|
        \;\le\; L_1 \, \frac{1}{x_u^\mathrm{in}} \| \tilde{x}^\mathrm{in} - x^\mathrm{in} \|,
\end{equation}
as well as the distance in the center direction,
\begin{equation}\label{eqXcInitial}
  | \tilde{x}_c(t) - x_c(t) | 
  \;\le\; d_c(t) + | \tilde{x}_c(0) - \tilde{x}_c^\mathrm{in} | + | \tilde{x}_c^\mathrm{in} - x_c^\mathrm{in} |
  \;\le\; d_c(t) + L_2 \| \tilde{x}^\mathrm{in} - x^\mathrm{in} \|.
\end{equation}
The constants $L_1, L_2>0$ are uniform in $x$, $\tilde{x}$, $\varepsilon$.

Vector field (\ref{eqSplitLocalSystem}) yields the estimate
\[
\begin{array}{rcl}
 (\tilde{x}_{s,ss} - x_{s,ss})' \;=\; -\mu_{s,ss}(\tilde{x}_c)(\tilde{x}_{s,ss} - x_{s,ss}) + R
\end{array}
\]
with remainder term
\[
\begin{array}{rcl}
\|R\| &\le & \|\mu_{s,ss}(\tilde{x}_c) - \mu_{s,ss}(x_c)\|\,\|x_{s,ss}\| 
\\ && ~       + \varepsilon \|\tilde{g}_{s,ss}(\tilde{x})\|\,\|\tilde{x}_{s,ss}-x_{s,ss}\| 
               + \varepsilon \|\tilde{g}_{s,ss}(\tilde{x}) - \tilde{g}_{s,ss}(x)\|\,\|x_{s,ss}\|
\\ &\le& \varepsilon C | \tilde{x}_c - x_c |\,\|x_{s,ss}\| + \varepsilon C \| \tilde{x}_{s,ss} - x_{s,ss} \| 
+ \varepsilon C \| \tilde{x} - x \|\,\|x_{s,ss}\|
\\ &=& \varepsilon C ( | \tilde{x}_c - x_c |\,\|x_{s,ss}\| + \| \tilde{x}_{s,ss} - x_{s,ss} \| ). 
\end{array}
\]
Here, $C$ is the uniform bound on $\tilde{g}$ and uniform Lipschitz-bound on $\tilde{g}$ and $\mu_{s,ss}$, 
see (\ref{eqGeneralHotBounds}), (\ref{eqGeneralHotLipschitzBounds}) and (\ref{eqEigenvalueGradients}).

In the last step we have used 
\[
\| \tilde{x} - x \|= | \tilde{x}_c - x_c |+ \| \tilde{x}_{s,ss} - x_{s,ss} \|
\]
since $| \tilde{x}_u - x_u |=0$ due to our choice (\ref{eqTimeInitialSegment}).

For arbitrary $0 < \alpha < 1$ we have $\alpha\mu_s < \mu_s(\tilde{x}_c) < \mu_{ss}(\tilde{x}_c)$ provided $\varepsilon_0$ is chosen accordingly. Hence
\[
\begin{array}{rcl}
d_s(t)' &\le& \| \tilde{x}_{s,ss}(t) - x_{s,ss}(t) \|'
\\ &\le&
   - \alpha\mu_s d_s(t) + \varepsilon C | \tilde{x}_c(t) - x_c(t) |\,\|x_{s,ss}(t)\|.
\end{array}
\]
Analogous Lipschitz estimates on $\tilde{g}_c$ yield
\[
\begin{array}{rcl}
d_c(t)' &=& | \tilde{x}_c(t)' - x_c(t)'|
\\ &\le& \varepsilon | x_u | \, \|\tilde{g}_c(\tilde{x})\|\,\|\tilde{x}_{s,ss}-x_{s,ss}\| 
         + \varepsilon | x_u | \, \|\tilde{g}_c(\tilde{x}) - \tilde{g}_c(x)\|\,\|x_{s,ss}\|
\\ &\le& \varepsilon C | x_u(t) | \, \left( \| \tilde{x}_{s,ss}(t) - x_{s,ss}(t) \| + | \tilde{x}_c(t) - x_c(t) |\,\|x_{s,ss}(t)\| \right )
\\ &=& \varepsilon C \exp(\mu_u t) | x_u^\mathrm{in} | \, \left( d_s(t) + | \tilde{x}_c(t) - x_c(t) |\,\|x_{s,ss}(t)\| \right ).
\end{array}
\]
Using estimate (\ref{eqXcInitial}), the exponential bounds on $x_{s,ss}$ of lemma \ref{thLocalBounds}, 
and $\|x_{s,ss}^\mathrm{in}\| < 2$, we obtain 
\begin{equation}\label{eqDsDcEvolution}
\begin{array}{rcl}
d_s(t)' &\le& - \alpha\mu_s d_s(t) 
   + 2 \varepsilon C \exp(-\alpha \mu_s t)
     \left( d_c(t) + L_2 \| \tilde{x}^\mathrm{in} - x^\mathrm{in} \| \right),
\\
d_c(t)' &\le& \varepsilon C \exp(\mu_u t) | x_u^\mathrm{in} | \, d_s(t) 
\\ && ~
   + 2 \varepsilon C \exp(-\alpha(\mu_s-\mu_u) t) | x_u^\mathrm{in} |
     \left( d_c(t) + L_2 \| \tilde{x}^\mathrm{in} - x^\mathrm{in} \| \right).
\end{array}
\end{equation}
The Gronwall estimate of $\exp(\alpha\mu_s t)d_s(t)$,
\begin{equation}\label{eqDsGronwall}
 d_s(t) \;\le\; \exp(-\alpha\mu_s t) 
   \left( d_s(0) + 2 \varepsilon C t L_2 \| \tilde{x}^\mathrm{in} - x^\mathrm{in} \| + 2 \varepsilon C \int_0^t d_c(s) \diff s \right),
\end{equation}
can be plugged into the estimate of $d_c$:
\[
\begin{array}{rcll}
d_c(t)' &\le& 
  2 \varepsilon C \exp(-\alpha(\mu_s-\mu_u) t) | x_u^\mathrm{in} |
  & \!\!\!\!\displaystyle \Big( d_c(t) + \varepsilon C \int_0^t d_c(s) \diff s 
\\ &&& ~~ + d_s(0) + (\varepsilon C t + 1) L_2 \| \tilde{x}^\mathrm{in} - x^\mathrm{in} \| \Big)
\end{array}
\]
Let $\| \tilde{x}^\mathrm{in} - x^\mathrm{in} \|$ be an upper bound on $d_c$, i.e.~assume
\begin{equation}\label{eqDcUpperBound}
 \| \tilde{x}^\mathrm{in} - x^\mathrm{in} \| \;>\; \sup_{0 \le s \le t} d_c(s).
\end{equation}
As $d_c(0)=0$, this assumption is valid on an initial segment.
Then 
\[
 d_c(t)' \;\le\; 
  2 \varepsilon C \exp(-\alpha(\mu_s-\mu_u) t) | x_u^\mathrm{in} |
  \left( d_c(t) + d_s(0) + (\varepsilon C t + 1) L_3 \| \tilde{x}^\mathrm{in} - x^\mathrm{in} \| \right),
\]
with $L_3 = L_2+1$. Gronwall inequality yields
\[
 d_c(t) \;\le\; \xi(t) + \int_0^t \xi(s)\eta(s) \exp\left(\int_s^t \eta(\sigma) \diff\sigma\right) \diff s
\]
with
\[
\begin{array}{rcl}
\eta(t) &=& 2 \varepsilon C \exp(-\alpha(\mu_s-\mu_u) t) | x_u^\mathrm{in} |
\\ &\le&
  2 \varepsilon C \exp(-\alpha(\mu_s-\mu_u) t),
\\[1ex]
\displaystyle \int_s^t \eta(\sigma) \diff\sigma &\le& \displaystyle \frac{2 \varepsilon C}{\alpha(\mu_s-\mu_u)}, \qquad 0 \le s \le t,
\\[2ex]
\xi(t) &=& \displaystyle 
  2 \varepsilon C | x_u^\mathrm{in} | \int_0^t \exp(-\alpha(\mu_s-\mu_u) s)
  \left( d_s(0) + (\varepsilon C s + 1) L_3 \| \tilde{x}^\mathrm{in} - x^\mathrm{in} \| \right) \diff s,
\\ &\le&
  \varepsilon \frac{2 C}{\alpha(\mu_s-\mu_u)} | x_u^\mathrm{in} | \left( d_s(0) + L_3 \| \tilde{x}^\mathrm{in} - x^\mathrm{in} \| \right)
  + \varepsilon^2 \frac{2 C^2}{\alpha^2(\mu_s-\mu_u)^2} | x_u^\mathrm{in} | L_3 \| \tilde{x}^\mathrm{in} - x^\mathrm{in} \|
\\ &\le&
  \varepsilon L_4 | x_u^\mathrm{in} | \left( d_s(0) + \| \tilde{x}^\mathrm{in} - x^\mathrm{in} \| \right)
\\ &\le&
  \varepsilon L_4 ( L_1+1 ) \| \tilde{x}^\mathrm{in} - x^\mathrm{in} \|.
\end{array}
\]
The second last inequality needs $\varepsilon_0$ chosen small enough, $L_4$ is a uniform constant. 
The last inequality uses the initial estimate (\ref{eqDsInitial}) and $|x_u^\mathrm{in}| \le 1$. Finally
\begin{equation}\label{eqDcBound}
\begin{array}{rcl}
  d_c(t) &\le& \displaystyle
  \xi(t) + \int_0^t \xi(s)\eta(s) \exp\left(\int_s^t \eta(\sigma) \diff\sigma\right) \diff s
\\ &\le&  \displaystyle
  \varepsilon L_4 ( L_1+1 ) \| \tilde{x}^\mathrm{in} - x^\mathrm{in} \|   
  \left( 1 + 2\varepsilon C \exp({\textstyle\frac{2 \varepsilon C}{\alpha(\mu_s-\mu_u)}}) \int_0^t \exp(-\alpha(\mu_s-\mu_u) s)   \diff s \right)
\\
  &\le& \varepsilon L_5 \| \tilde{x}^\mathrm{in} - x^\mathrm{in} \|,
\end{array}
\end{equation}
with $L_5$ independent of $\varepsilon$. Choose $\varepsilon_0$ small enough. 
Then $\varepsilon L_5 \ll 1$ and assumption (\ref{eqDcUpperBound}) remains valid up to the out-section.
This yields the first claim of the theorem.

Estimates (\ref{eqDcBound}) and (\ref{eqDsGronwall}) combine to
\[
 d_s(t) \;\le\; \exp(-\alpha\mu_s t) \left( d_s(0) + \varepsilon t L_6 \| \tilde{x}^\mathrm{in} - x^\mathrm{in} \| \right).
\]
Thus
\[
\begin{array}{rcl}
 d_s^\mathrm{out} &=& d_s(\textstyle\frac{1}{\mu_u} \ln \frac{1}{|x_u^\mathrm{in}|}) 
\\[1ex] &\le& \displaystyle
    |x_u^\mathrm{in}|^{\alpha\mu_s/\mu_u} 
    \left( L_1 \frac{\|\tilde{x}^\mathrm{in}-x^\mathrm{in}\|}{|x_u^\mathrm{in}|} 
           + \frac{\varepsilon L_6}{\mu_u} \| \tilde{x}^\mathrm{in} - x^\mathrm{in} \| \ln \frac{1}{|x_u^\mathrm{in}|} \right)
\\ &\le& \displaystyle
   |x_u^\mathrm{in}|^{\alpha\mu_s/\mu_u-1} L_s \| \tilde{x}^\mathrm{in} - x^\mathrm{in} \|.
\end{array}
\]
Note that $\beta = \alpha\mu_s/\mu_u - 1 > 0$ for $\alpha$ chosen close enough to $1$. This finishes the proof.
\end{proof}


\section{The return map}
\label{secReturnMap}

In this section we define a return map for trajectories near the primary heteroclinic 
orbit $q(t)$ considered as a homoclinic orbit relative to the equivariance.
In this way a map on the in-section is defined which induces a transformation of 
Lipschitz-curves on this section.
We first prove Lipschitz- and cone properties of the return map.
The homoclinic orbit $q(t)$ corresponds to a fixed point of this map.
The stable set of the equilibrium under the return map is then constructed
as the fixed point of the associated graph transformation.
This yields the set of trajectories with $\alpha$-limit dynamics following the 
period-3 heteroclinic cycle resp. the relative homoclinic orbit $q(t)$.

Locally near $q_-$ we can use the general result on the local passage of the previous section and apply it to the Bianchi system.
The time-reversed Bianchi flow satisfies assumptions (\ref{eqGeneralEqLine}, \ref{eqGeneralCrudeLinearization}).
The out-section is then diffeomorphically mapped back to the in-section by the time-reversed (global) Bianchi flow
after reduction to the orbit space of the equivariance group,
\begin{equation}\label{eqGlobalExcursionMap}
  (x_s^\mathrm{out}, x_{ss}^\mathrm{out}, x_c^\mathrm{out}) \longmapsto 
  (x_u^\mathrm{in}, x_s^\mathrm{in}, x_c^\mathrm{in}) = 
    \Psi^\mathrm{glob} (x_s^\mathrm{out}, x_{ss}^\mathrm{out}, x_c^\mathrm{out})
\end{equation}

In particular the line $\{x_{ss}^\mathrm{out}=x_{s}^\mathrm{out}=0\}$ is mapped to the line $\{x_{u}^\mathrm{in}=x_{s}^\mathrm{in}=0\}$ by the Kasner map. 
Thus $\Psi^\mathrm{glob}$ expands the $x_c$-coordinate uniformly in a neighborhood of the 3-periodic cycle $q(t)$.

Combining both maps (the local passage and the global excursion) we obtain a return map 
\begin{equation}\label{eqReturnMap}
\Psi \,:=\; \Psi^\mathrm{glob} \circ \Psi^\mathrm{loc} \,:\; \overline\Sigma^\mathrm{in} \longrightarrow \overline\Sigma^\mathrm{in},
\end{equation}
%
%
see also figure \ref{figGlobalReturn}.
Note, that due to expansion of the center component $x_c$ under the return map, 
the domain of definition of this map is in fact a subset of $\overline\Sigma^\mathrm{in}$.

\begin{figure}
\setlength{\unitlength}{0.01\textwidth}
\centering
\begin{picture}(85,45)(-5,-5)
\thicklines
\put(  0,  0){\line(1,0){12}}\put( 20,  0){\line(1,0){15}}
\put(  0,  0){\line(0,1){16}}\put(  0, 20){\line(0,1){5}}
\put(-5,-2.5){\line(2,1){10}}
\put( -5, -5){\thinlines\line(1,1){10}}
\put(  5,  0){\vector(-1,0){1}}
\put(  6,  0){\vector(-1,0){1}}
\put(  7,  0){\vector(-1,0){1}}
\put(  5,  2.5){\vector(-2,-1){1}}
\put(  6,  3){\vector(-2,-1){1}}
\put( -5, -2.5){\vector(2,1){1}}
\put( -6, -3){\vector(2,1){1}}
\put(  0,  6){\vector(0,1){1}}
\thinlines
\put(20,0){%
\put( -8, -4){\line(2,1){16}}
\put(  8,  4){\line(0,1){8}}
\put( -8, -4){\line(0,1){8}}
\put( -8,  4){\line(2,1){16}}
\put(  8,  8){\makebox(0,0)[l]{$\,\Sigma^\mathrm{in}$}}
}
\thinlines
\put(0,20){%
\put( -8, -4){\line(2,1){16}}
\put(  8,  4){\line(1,0){10}}
\put( -8, -4){\line(1,0){10}}
\put(  2, -4){\line(2,1){16}}
\put( 13,  4.5){\makebox(0,0)[b]{$\Sigma^\mathrm{out}$}}
}
\thinlines
\qbezier(0,25)(0,40)(35,40)
\qbezier(35,40)(80,40)(80,25)
\qbezier(80,25)(80,0)(35,0)
\thicklines
\put(25,5){\line(1,0){10}}
\put(4,21){\line(0,1){4}}
\qbezier(14,5)(4,5)(4,17)
\put(5,11){\vector(-1,6){0.1}}
\put(25,5){\circle*{0.5}}
\put(4,21){\circle*{0.5}}
\put(7,9){\makebox(0,0)[bl]{$\Psi^\mathrm{loc}$}}
\put(25,5){\makebox(0,0)[br]{$x^\mathrm{in}$}}
\put(4,21){\makebox(0,0)[tl]{$x^\mathrm{out}$}}
\qbezier(4,25)(4,35)(35,35)
\qbezier(35,35)(75,35)(75,25)
\qbezier(75,25)(75,2)(32,2)
\put(22,2){\line(1,0){10}}
\put(22,2){\circle*{0.5}}
\put(70,30){\makebox(0,0)[tr]{$\Psi^\mathrm{glob}$}}
\put(73,35){\makebox(0,0)[bl]{$q(t)$}}
\put(22,2){\makebox(0,0)[br]{$\Psi(x^\mathrm{in})$}}
\end{picture}
\caption{\label{figGlobalReturn}The return map. $\Psi = \Psi^\mathrm{glob} \circ \Psi^\mathrm{loc}$.}
\end{figure}

The heteroclinic cycle $q(t)$ yields a homoclinic orbit 
after reduction to the orbit space of the equivariance group. 
It constitutes the unique fixed point $(0,0,0)$ of the return map $\Psi$.
The Kasner circle and the attached heteroclinic cap is represented by the set 
$\{x_u^\mathrm{in} = x_s^\mathrm{in} = 0\}$. The return map becomes singular on the stable manifold of the Kasner circle. 
In particular, we define on this singular set
\begin{equation}\label{eqSingularLineReturn}
\Psi (0, x_s^\mathrm{in}, x_c^\mathrm{in}) = (0,0,\Phi(x_c^\mathrm{in})),
\end{equation}
where $\Phi$ is the Kasner map.

Note that the Bianchi system (\ref{eq5dimBianchi}) preserves the signs of $N_1, N_2, N_3, \Omega$.
Although the local analysis of the previous section did not use this special structure, it is compatible with it.
To study Bianchi-IX solutions close to the 3-cycle of heteroclinics we therefore consider only the part 
of the in-section with $x_u\ge 0$, $x_s\ge 0$. This part is mapped to the in-section, again with $x_u\ge 0$, $x_s\ge 0$.

In complete analogy one can treat Bianchi-VIII solutions close to a 3-cycle, 
although this 3-cycle is then composed of orbits inside heteroclinic caps of different sign. 
Indeed, each heteroclinic orbit to the Kasner circle has a mirror image under sign reversal of $N_k$.
Each of the 3 heteroclinic orbits of the cycle can be given a sign, the choice of signs determines whether the cycle 
lies in the boundary of Bianchi-VIII or Bianchi-IX domains.

In fact, looking at a section to a given heteroclinic orbit of the cycle, thus fixing the sign of the first orbit,
we have 4 quadrants in this section corresponding to the 4 remaining choices of signs of the 2 remaining orbits.
In each of the quadrants, theorem \ref{thStableSet} will yield a local Lipschitz manifold.

Although the analysis of this section also holds for a 3-cycle in the boundary of Bianchi-VIII domains, 
we keep the notation corresponding to Bianchi-IX, i.e.~all components $x_s, x_{ss}, x_u$ are and remain non-negative.
These components correspond to the variables $N_1,N_2,N_3$ in the Bianchi system, the identification permutes along the cycle.
The component $x_c$ correspond to the variables $\Sigma_\pm$.

In the following, we first prove suitable cone conditions for the return map $\Psi$ and 
then apply almost standard graph transformation to obtain the claimed Lipschitz manifold. 
The nonstandard part of that graph transformation includes the singular line $x_u^\mathrm{in}=0$ of infinite contraction.

Note the order of fixing the rescaling parameters:
First $\varepsilon_0$ resp.~$\varepsilon$ is fixed small enough to yield our estimates of the 
local passage $\Psi^\mathrm{loc}$ with small Lipschitz constants, in particular $\varepsilon L_c \ll 1$. 
This amounts to a choice of the sections $\Sigma^\mathrm{in}$ and $\Sigma^\mathrm{out}$ in the original (unscaled) coordinates
and also fixes the global excursion map $\Psi^\mathrm{glob}$.
Then a sufficiently small upper bound for $x_u^\mathrm{in}$ is chosen, i.e.~$\Psi$ is restricted to a smaller section
\begin{equation}\label{eqReducedInSection}
\tilde\Sigma^\mathrm{in} = \{x\in\overline\Sigma^\mathrm{in} \;|\; 0 \le x_u \le \delta,\; 0 \le x_s \le \delta,\; |x_c| < \delta \}
\end{equation}
This makes the contraction of the local passage as strong as we like without changing $\Psi^\mathrm{loc}$, $\Psi^\mathrm{glob}$.
It also ensures that trajectories of interest stay close to the Kasner caps of heteroclinic orbits, 
thus we ensure that the global excursion $\Psi^\mathrm{glob}$ on the domain of interest is as close to the Kasner map as we like.
It also ensures that all non-singular trajectories in this domain indeed return to the in-section $\overline\Sigma^\mathrm{in}$.

\begin{lem}\label{thLipschitzReturnMap}
The return map (\ref{eqReturnMap}) is Lipschitz continuous. 
Furthermore, there exist $\varepsilon > 0$, $\delta>0$, $0<\sigma<1$, $K_{u,s} > 1$, and $K_c > (1-\sigma^2)^{-1} > 1$,
such that the following cone conditions hold for
$\Psi=\Psi^\mathrm{glob} \circ \Psi^\mathrm{loc} : \tilde\Sigma^\mathrm{in} \to \overline\Sigma^\mathrm{in}$.

Here $\overline\Sigma^\mathrm{in}$ is the closed in-section (\ref{eqInSection}) corresponding to the choice of $\varepsilon$, 
and $\tilde\Sigma^\mathrm{in}$ a suitable subset (\ref{eqReducedInSection}).

The cones are defined for $x\in\tilde\Sigma^\mathrm{in}$ as
\begin{equation}\label{eqCones}
\begin{array}{rcl}
C_x^\mathrm{c} &=& \{\tilde{x}\in\tilde\Sigma^\mathrm{in}\;|\; \|\tilde{x}_{u,s}-x_{u,s}\| \le \sigma |\tilde{x}_c-x_c| \},
\\
C_x^\mathrm{u,s} &=& \{\tilde{x}\in\tilde\Sigma^\mathrm{in}\;|\; |\tilde{x}_c-x_c| \le \sigma \|\tilde{x}_{u,s}-x_{u,s}\| \}.
\end{array}
\end{equation}
The cone conditions are
\begin{enumerate}
 \item[(i)] Invariance: $\Psi(C_x^\mathrm{c}) \cap \tilde\Sigma^\mathrm{in} \subset (\mathrm{int\,} C_{\Psi x}^\mathrm{c}) \cup \{\Psi x\}$ and 
   $\Psi^{-1}(C_{\Psi x}^\mathrm{u,s}) \cap \tilde\Sigma^\mathrm{in} \subset (\mathrm{int\,} C_x^\mathrm{u,s}) \cup \{x\}$;
 \item[(ii)] Contraction \& Expansion: For all $\tilde{x}\in C_x^\mathrm{c}$ we have expansion in the center direction: 
   $|(\Psi\tilde{x})_c-(\Psi x)_c| \ge K_c |\tilde{x}_c-x_c|$
   and for all $\Psi\tilde{x}\in C_{\Psi x}^\mathrm{u,s}$ we have contraction in the transverse directions:
   $\|\tilde{x}_{u,s}-x_{u,s}\| \ge K_{u,s} \|(\Psi\tilde{x})_{u,s}-(\Psi x)_{u,s}\|$.
\end{enumerate}
They hold for all, $x,\tilde{x},\Psi x,\Psi\tilde{x} \in \tilde\Sigma^\mathrm{in}$. See also figure \ref{figConeConditions}.
\end{lem}

\begin{figure}
\setlength{\unitlength}{0.009\textwidth}
\centering
\begin{picture}(100,40)(-50,-5)
\put(-45,  0){%
\put(  0, -5){\vector(0,1){38}}\put(0,35){\makebox(0,0)[t]{$x_c$}}
\put( -5,  0){\vector(1,0){35}}\put(35,0){\makebox(0,0)[r]{$x_{u,s}$}}
\put( 22, 12){%
\put(  0,  0){\rotatebox{-6.59}{\makebox(0,0)[bl]{\color[rgb]{0.8,0.8,0.8}\rule{14\unitlength}{5\unitlength}}}}
\put(  0,  0){\rotatebox{-6.59}{\makebox(0,0)[tr]{\color[rgb]{0.8,0.8,0.8}\rule{14\unitlength}{5\unitlength}}}}
\put(  0,  0){\rotatebox{ 6.59}{\makebox(0,0)[bl]{\color[rgb]{1.0,1.0,1.0}\rule{14\unitlength}{5\unitlength}}}}
\put(  0,  0){\rotatebox{ 6.59}{\makebox(0,0)[tr]{\color[rgb]{1.0,1.0,1.0}\rule{14\unitlength}{5\unitlength}}}}
\put( 13,  0){\makebox(0,0)[l]{\color[rgb]{1.0,1.0,1.0}\rule{5\unitlength}{15\unitlength}}}
\put(-13,  0){\makebox(0,0)[r]{\color[rgb]{1.0,1.0,1.0}\rule{5\unitlength}{15\unitlength}}}
\put(  0,  0){\circle*{1}}
\put( -5,-10){\line(1,2){10}}
\put(  5,-10){\line(-1,2){10}}
\put(-13, -1.5){\line(26,3){26}}
\put(-13,  1.5){\line(26,-3){26}}
\put(-10, -5){\line(2,1){20}}
\put(-10,  5){\line(2,-1){20}}
\put(  0, -2){\makebox(0,0)[t]{$x$}}
\put( -4.2,  8){\makebox(0,0)[bl]{$C^\mathrm{c}_x$}}
\put(  8,  4.2){\makebox(0,0)[tl]{$C^\mathrm{u,s}_x$}}
\put(-13,  0){\makebox(0,0)[r]{$\Psi^{-1}(C^\mathrm{u,s}_{\Psi x})$}}
}
}
\put( 15,  0){%
\put(  0, -5){\vector(0,1){37}}\put(0,35){\makebox(0,0)[t]{$\Psi(x)_\mathrm{c}$}}
\put( -5,  0){\vector(1,0){30}}\put(35,0){\makebox(0,0)[r]{$\Psi(x)_\mathrm{u,s}$}}
\put( 12, 25){%
\put(  0,  0){\rotatebox{83.41}{\makebox(0,0)[bl]{\color[rgb]{0.8,0.8,0.8}\rule{14\unitlength}{5\unitlength}}}}
\put(  0,  0){\rotatebox{83.41}{\makebox(0,0)[tr]{\color[rgb]{0.8,0.8,0.8}\rule{14\unitlength}{5\unitlength}}}}
\put(  0,  0){\rotatebox{96.59}{\makebox(0,0)[bl]{\color[rgb]{1.0,1.0,1.0}\rule{14\unitlength}{5\unitlength}}}}
\put(  0,  0){\rotatebox{96.59}{\makebox(0,0)[tr]{\color[rgb]{1.0,1.0,1.0}\rule{14\unitlength}{5\unitlength}}}}
\put(  0, 13){\makebox(0,0)[b]{\color[rgb]{1.0,1.0,1.0}\rule{15\unitlength}{5\unitlength}}}
\put(  0,-13){\makebox(0,0)[t]{\color[rgb]{1.0,1.0,1.0}\rule{15\unitlength}{5\unitlength}}}
\put(  0,  0){\circle*{1}}
\put( -5,-10){\line(1,2){10}}
\put(  5,-10){\line(-1,2){10}}
\put(-1.5,-13){\line(3,26){3}}
\put( 1.5,-13){\line(-3,26){3}}
\put(-10, -5){\line(2,1){20}}
\put(-10,  5){\line(2,-1){20}}
\put( -2,  0){\makebox(0,0)[r]{$\Psi(x)$}}
\put( -4.2,  8){\makebox(0,0)[bl]{$C^\mathrm{c}_{\Psi x}$}}
\put(  8,  4.2){\makebox(0,0)[tl]{$C^\mathrm{u,s}_{\Psi x}$}}
\put(  0,-13){\makebox(0,0)[t]{$\Psi(C^\mathrm{c}_x)$}}
}
}
\put(0,0){\makebox(0,0)[b]{\begin{minipage}{16\unitlength}\centering
{\Huge$\stackrel{\mathbf{\Psi}}{\longrightarrow}$}\\
Expansion of $x_c$
Contraction of $x_{u,s}$
\end{minipage}}}
\end{picture}
\caption{\label{figConeConditions}Cone properties of the return map $\Psi$.}
\end{figure}
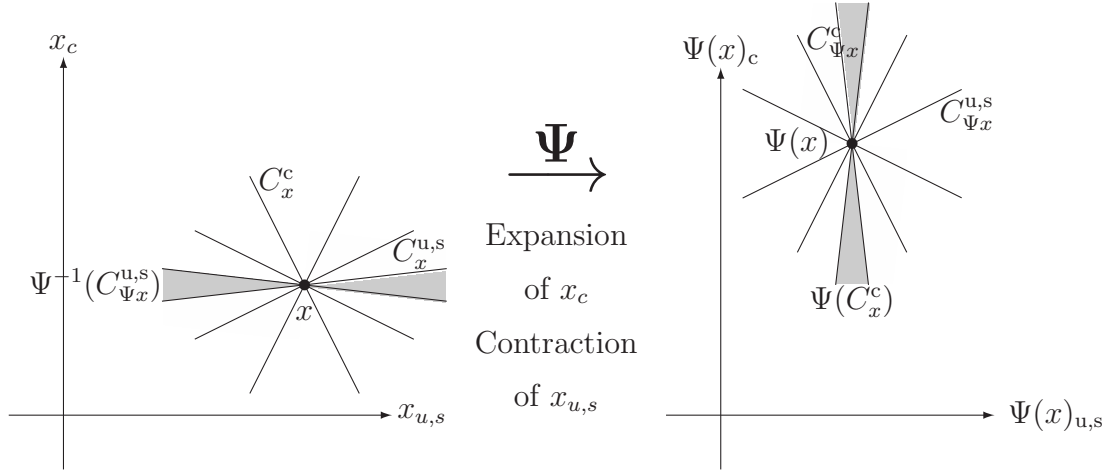

\begin{proof}
Lipschitz continuity of the return map $\Psi$ follows directly from Lipschitz continuity of the local passage $\Psi^\mathrm{loc}$,
see theorem \ref{thLipschitzLocalMap}, as the global excursion $\Psi^\mathrm{glob}$ is smooth.

The cone conditions require the expansion in $x_c$-direction along the Kasner circle induced by the Kasner map.
On the singular line, the global excursion is given by the Kasner map $\Phi$,
\[
 \Psi^\mathrm{glob} (x_s^\mathrm{out} = 0, x_{ss}^\mathrm{out} = 0, x_c^\mathrm{out}) 
  \;=\; (x_u^\mathrm{in}=0, x_s^\mathrm{in}=0, \Phi(x_c^\mathrm{out})).
\]
Indeed, the local passage does not change $x_c$ on the singular line, so this notation 
is consistent with (\ref{eqSingularLineReturn}).
Therefore we can write
\[
 \Psi^\mathrm{glob} (x_s, x_{ss}, x_c) \;=\; (0,0,\Phi(x_c)) + \tilde{\Psi}^\mathrm{glob} (x_s, x_{ss}, x_c) x_{s,ss},
\]
with a smooth $(3\times 2)$-matrix $\tilde{\Psi}^\mathrm{glob} (x_s, x_{ss}, x_c)$ and vector $x_{s,ss}$.
We obtain the following Lipschitz estimates
\[\begin{array}{rcl}
&&\hspace{-4em}
\| \Psi^\mathrm{glob} (\tilde{x}) - \Psi^\mathrm{glob} (x) - (0,0,\Phi(\tilde{x}_c) - \Phi(x_c)) \|
\\ &=& \left\| \tilde{\Psi}^\mathrm{glob}(\tilde{x}) \left(\tilde{x}_{s,ss} - x_{s,ss}\right) 
          + \left(\tilde{\Psi}^\mathrm{glob}(\tilde{x}) - \tilde{\Psi}^\mathrm{glob}(x)\right)x_{s,ss} \right\|
\\ &\le& C^\mathrm{glob} \left( \|\tilde{x}_{s,ss} - x_{s,ss}\| + \|\tilde{x}-x\| \|x_{s,ss}\| \right),
\end{array}
\]
with uniform bound $C^\mathrm{glob} > \|\tilde{\Psi}^\mathrm{glob}\|, \|D\tilde{\Psi}^\mathrm{glob}\|$ on 
the sup-norm of the higher order terms of the global excursion map and their derivatives. 
Note that $\tilde{x}$, $x$ lie on the out-section and $\Psi^\mathrm{glob}$ is a diffeomorphism 
from the out-section to its image on the in-section.

Using $\Psi^\mathrm{loc}(\tilde{x})$ and $\Psi^\mathrm{loc}(x)$ instead of $\tilde{x}$ and $x$ we get a similar estimate 
for the return map $\Psi(x)=\Psi^\mathrm{glob}(\Psi^\mathrm{loc}(x))$:
\begin{equation}\label{eqEstimateReturnMap}
\begin{array}{rcl}
&&\hspace{-4em}
\| \Psi(\tilde{x}) - \Psi(x) -(0,0,\Phi(\Psi^\mathrm{loc}(\tilde{x})_c) - \Phi(\Psi^\mathrm{loc}(x)_c)) \|
\\ &=& \| \Psi^\mathrm{glob}(\Psi^\mathrm{loc}(\tilde{x})) - \Psi^\mathrm{glob}(\Psi^\mathrm{loc}(x)) 
            - (0,0,\Phi(\Psi^\mathrm{loc}(\tilde{x})_c) - \Phi(\Psi^\mathrm{loc}(x)_c)) \|
\\ &\le& C^\mathrm{glob} \left( \| \Psi^\mathrm{loc}(\tilde{x})_{s,ss} - \Psi^\mathrm{loc}(x)_{s,ss}\| +
                                \|\Psi^\mathrm{loc}(\tilde{x})-\Psi^\mathrm{loc}(x)\| \|\Psi^\mathrm{loc}(x)_{s,ss}\| \right)
\\ &\le& C^\mathrm{glob} \left( L_s |x_u|^\beta \| \tilde{x} - x \| 
                                + ( \varepsilon L_c + 1+L_s |x_u|^\beta ) \| \tilde{x} - x \| |x_u|^\beta \| x \| \right)
\\ &\le& C^\mathrm{return} |x_u|^\beta \| \tilde{x} - x \|.
\end{array}
\end{equation}
The second last inequality uses the estimates of the local passage of corollary \ref{thContinuousLocalMap} 
and theorem \ref{thLipschitzLocalMap} for the choice (w.l.o.g.) $0 \le \tilde{x}_u \le x_u$. 
Note that the estimates of theorem \ref{thLipschitzLocalMap} are used in the form
\begin{eqnarray*}
| \Psi^\mathrm{loc}(\tilde{x})_c - \Psi^\mathrm{loc}(x)_c| & \leq &  
\varepsilon L_c\|\tilde{x}-x\| + |\tilde{x}_c-x_c| \leq (\varepsilon L_c+1)\|\tilde{x}-x\|, \\
\| \Psi^\mathrm{loc}(\tilde{x})_{s,ss} - \Psi^\mathrm{loc}(x)_{s,ss}\| & \leq & 
L_s |x_u|^\beta  \| \tilde{x} - x \|.
\end{eqnarray*}
The constant $C^\mathrm{return}$ is uniform in $x$, $\tilde{x}$ in the in-section. 
Because $\beta >0$, we have an arbitrarily strong contraction on $\tilde\Sigma^\mathrm{in}$, 
i.e.~for $x_u < \delta$, if we choose $\delta$ small enough.

The Kasner map $\Phi$ in the original Bianchi system is expanding everywhere except at the Taub points $T_1,T_2,T_3$. 
In particular it is expanding on our in-section:
\[
  | \Phi(a) - \Phi(b) |    \;\ge\;   \tilde{K}_c |a-b|,
\]
for some uniform constant $\tilde{K}_c > 1$.

Now choose $K_c$ with $1 < K_c < \tilde{K_c}$, and $\sigma$ with $0 < \sigma < 1$ such that $K_c(1-\sigma^2) > 1$.
(The last relation is needed to obtain a contraction in theorem \ref{thStableSet}.)

Consider the cone in center direction with opening $\vartheta>0$, 
i.e.~$\|\tilde{x}_{u,s}-x_{u,s}\| \le \vartheta |\tilde{x}_c-x_c|$.
Then (\ref{eqEstimateReturnMap}) using the local Lipschitz estimate of theorem \ref{thLipschitzLocalMap} yield
\begin{equation}\label{eqEstimateExpansion}
\begin{array}{rcl}
|(\Psi(\tilde{x}) - \Psi(x))_c| 
&\ge& | \Phi(\Psi^\mathrm{loc}(\tilde{x})_c) - \Phi(\Psi^\mathrm{loc}(x)_c) | - C^\mathrm{return} |x_u|^\beta \| \tilde{x} - x \|
\\ &\ge& \tilde{K}_c | \Psi^\mathrm{loc}(\tilde{x})_c - \Psi^\mathrm{loc}(x)_c | - C^\mathrm{return} |x_u|^\beta \| \tilde{x} - x \|
\\ &\ge& \tilde{K}_c | \tilde{x}_c - x_c | - \varepsilon L_c \tilde{K}_c \| \tilde{x} - x \| - C^\mathrm{return} |x_u|^\beta \| \tilde{x} - x \|
\\ &\ge& \left( \tilde{K}_c - \left(\varepsilon L_c \tilde{K}_c + C^\mathrm{return} |x_u|^\beta \right)(1+\vartheta) \right) 
         | \tilde{x}_c - x_c |.
\end{array}
\end{equation}
For $\varepsilon$ and $\delta$ chosen small enough, using $x_u \le \delta$, we can achieve
\[
 K_c \; < \; \tilde{K}_c - \left(\varepsilon L_c \tilde{K}_c + C^\mathrm{return} |x_u|^\beta \right)(1+1/\sigma),
\]
yielding the expansion not only in the cone $C_x^\mathrm{c}$, with $\vartheta=\sigma<1$, 
but also outside the cone $C_x^\mathrm{u,s}$, with $\vartheta=1/\sigma$.

Furthermore, using again (\ref{eqEstimateReturnMap}), we see the invariance of the cones. 
Indeed, assume again $\|\tilde{x}_{u,s}-x_{u,s}\| \le \vartheta |\tilde{x}_c-x_c|$, then we have
\[
\begin{array}{rcl}
 \|\Psi(\tilde{x})_{u,s} - \Psi(x)_{u,s}\| 
   &\le& C^\mathrm{return} |x_u|^\beta \| \tilde{x} - x \|
\\ &\le& C^\mathrm{return} |x_u|^\beta (1+\vartheta) | \tilde{x}_c - x_c |
\\ &\le& C^\mathrm{return} |x_u|^\beta (1+\vartheta) K_c^{-1} | \Psi(\tilde{x})_c - \Psi(x)_c |.
\end{array}
\]
The last inequality uses the expansion in $x_c$, thus it is valid for $\vartheta \le 1/\sigma$.
We choose $\delta$ small enough such that $C^\mathrm{return} |x_u|^\beta K_c^{-1} < \sigma/(1+\sigma)$.
Due to the monotone increase of $\vartheta/(1+\vartheta)$ we also have 
$C^\mathrm{return} |x_u|^\beta K_c^{-1} < \vartheta/(1+\vartheta)$ for all $\vartheta \ge \sigma$.
Thus we obtain the cone invariance
\begin{equation}\label{eqEstimateCone}
 \|\Psi(\tilde{x})_{u,s} - \Psi(x)_{u,s}\| \;\le\; \vartheta | \Psi(\tilde{x})_c - \Psi(x)_c |
\end{equation}
for all $\sigma \le \vartheta \le 1/\sigma$.

The choice $\vartheta=\sigma$ yields (forward) invariance of the cone $C_x^\mathrm{c}$ and the choice
$\vartheta = 1/\sigma$ yields (backward) invariance of the cone $C_{\Psi x}^\mathrm{u,s}$.
Note that the cone invariances are in fact strict as claimed in the lemma.
The above estimates are strict inequalities for $x\ne\tilde{x}$.

Now consider the cone in transverse direction, that is  $\Psi(\tilde{x}) \in C_{\Psi x}^\mathrm{u,s}$,
which amounts to $|\Psi(\tilde{x})_c-\Psi(x)_c| \le \sigma \|\Psi(\tilde{x})_{u,s}-\Psi(x)_{u,s}\|$. 
We have already established invariance. 
Thus $|\tilde{x}_c-x_c| \le \sigma \|\tilde{x}_{u,s}-x_{u,s}\|$ and estimate (\ref{eqEstimateReturnMap}) yields
\[
\begin{array}{rcl}
 \|\Psi(\tilde{x})_{u,s} - \Psi(x)_{u,s}\| 
   &\le& C^\mathrm{return} |x_u|^\beta \| \tilde{x} - x \|
\\ &\le& C^\mathrm{return} |x_u|^\beta (1+\sigma) \| \tilde{x}_{u,s} - x_{u,s} \|.
\end{array}
\] 
This is the claimed contraction, $K_{u,s}^{-1} = C^\mathrm{return} \delta^\beta (1+\sigma) < 1$, 
for $\delta$ small enough.
\end{proof}

\begin{thm}\label{thStableSet}
The (local) stable set of $(0,0,0)$ under the return map $\Psi$ is given by
\[
\{ (x_u, x_s, x_c) \,|\, x_c = x_c(x_u, x_s) \}
\]
with a Lipschitz continuous function $x_c$. In particular, $x_c(0, x_s) = 0$.
\end{thm}

\begin{proof}
The idea of the proof is to define a graph transformation on the space of Lip\-schitz graphs $\{(x_u, x_s, x_c(x_u, x_s))\}$, 
figure \ref{figGraphTransform}, by the inverse return map $\Psi^{-1}$.
The cone invariance of the previous lemma will ensure that the Lipschitz property of the graph is preserved.
Due to the expansion/contraction conditions of the previous lemma, the graph transformation turns out to be a contraction
on the space of graphs. The fixed point of this contraction then yields the claim.


\begin{figure}
\setlength{\unitlength}{0.01\textwidth}
\centering
\begin{picture}(60,40)(-5,-5)
\put( -5,  0){\line(1,0){60}}\put(55,-0.5){\makebox(0,0)[tr]{$x_u$}}
\put(  0, -5){\line(0,1){40}}\put( 0,35){\makebox(0,0)[tr]{$x_c\,$}}
\put(-3.75,-5){\line(3,4){30}}\put(26.25,35){\makebox(0,0)[tr]{$x_s\;\;$}}
\linethickness{2pt}
\put(  0,  0){\line(3,4){21}}
\put(  0,  0){\line(1,0){35}}
\cbezier(21,28)(30,18)(45,25)(55,35)
\cbezier(35,0)(40,30)(50,5)(55,35)
\end{picture}
\caption{\label{figGraphTransform}Lipschitz graph $x_c=x_c(x_s,x_u)$ with fixed line $x_c(x_s,0)\equiv 0$. 
The graph transformation given by the inverse return map $\Psi^{-1}$ expands $x_{u,s}$ and contracts $x_c$.}
\end{figure}
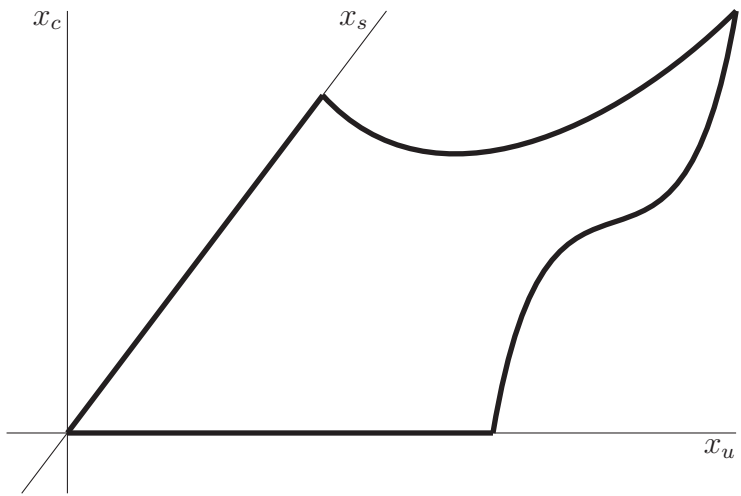

To make this idea precise, consider the Banach space of Lipschitz functions 
\[
\begin{array}{rcrll}
 X &=& \{ &
             \zeta\,:\;[0,\delta]^2\to[-\delta,\delta],\;x_{u,s}=(x_u,x_s) \mapsto x_c=\zeta(x_{u,s})
\\ &&& \mbox{such that } \zeta(x_s,0)\equiv 0,\; \mathrm{Lip}(\zeta)\le\sigma                          & \}
\end{array}
\]
with sup-norm. The parameters $\delta,\sigma < 1$ correspond to those of lemma \ref{thLipschitzReturnMap}.

Define a map $G:X\to X$ as graph transformation, i.e.~$\mathrm{graph}(G\zeta) := \Psi^{-1}\mathrm{graph}(\zeta)$.
More precisely
\[
\begin{array}{rcll}
G\zeta\left(\left(\Psi^{-1}(x_{u,s},\zeta(x_{u,s}))\right)_{u,s}\right) 
  &:=& \left(\Psi^{-1}(x_{u,s},\zeta(x_{u,s}))\right)_c,
  & \mbox{for } x_{u,s} \ne 0,
\\
G\zeta(0,x_s) &:=& 0.
  & \mbox{for arbitrary } x_s,
\end{array}
\]
The first equation implicitly assumes that $(x_{u,s},\zeta(x_{u,s}))$ has a pre-image under $\Psi$ and that it lies in the domain.
The second equation just gives the pre-image of the origin under $\Psi$.
Note the restriction to non-negative $x_u,x_s$ consistent with the structure of the
Bianchi system that preserves the signs of the spatial curvature variables $N_k$.
In fact due to these constraints the boundary $x_s=0$ becomes singular after one excursion, $\Psi(x_u,x_s=0,x_c)_u = 0$.
This fact is not explicitly used in the following but justifies the choice of the domain.

We will prove the following claims
\begin{enumerate}
\item[(i)] domain of definition:
  for all $\zeta\in X$ and $x_{u,s} \in (0,\delta]\times[0,\delta]$ 
  there exists $\tilde{x}_{u,s} \in [0,\delta]^2\setminus\{0\}$, such that
  $\left(\Psi^{-1}(\tilde{x}_{u,s},\zeta(\tilde{x}_{u,s}))\right)_{u,s} = x_{u,s}$.
\item[(ii)] well-definedness:
  for all $\zeta\in X$ and $x_{u,s}, \tilde{x}_{u,s} \in [0,\delta]^2\setminus\{0\}$ the following holds.
  If $\left(\Psi^{-1}(x_{u,s},\zeta(x_{u,s}))\right)_{u,s} = \left(\Psi^{-1}(\tilde{x}_{u,s},\zeta(\tilde{x}_{u,s}))\right)_{u,s} \in [0,\delta]^2$
  then already $x_{u,s} = \tilde{x}_{u,s}$.
\end{enumerate}
Conditions (i) and (ii) yield a well defined function $G\zeta$ for every $\zeta\in X$. 
\begin{enumerate}
\item[(iii)] Lipschitz property:
  for all $\zeta\in X$ the function $G\zeta$ is again Lipschitz continuous with $\mathrm{Lip}(G\zeta)\le \sigma$.
\item[(iv)] contraction: The exists a constant $0<\kappa<1$ such that
  for all $\zeta,\tilde{\zeta} \in X$ the estimate 
  $\| G\tilde{\zeta} - G\zeta \|_\mathrm{sup} \le \kappa \| \tilde{\zeta} - \zeta \|_\mathrm{sup}$ holds.
\end{enumerate}
Conditions (i)--(iii) prove that the graph transformation $G$ 
indeed maps Lipschitz functions in $X$ to Lipschitz functions in $X$.
Condition (iv) provides a contraction.
If all four conditions hold, then by contraction-mapping theorem there is a unique fixed point, 
i.e.~a Lipschitz function $\zeta^*\in X$ with $G\zeta^* = \zeta^*$. 
Its graph is a locally invariant manifold under $\Psi$. It is also the stable set of the origin due 
to the cone conditions of lemma \ref{thLipschitzReturnMap}. This yields the claim of the theorem.
Therefore it remains to prove (i)--(iv):

(i) Let $\zeta\in X$ and $x_{u,s} \in (0,\delta]\times[0,\delta]$ be given. 
The straight line $\{x_{u,s}\} \times [-\delta,\delta]  \subset \tilde\Sigma^\mathrm{in}$ is contained in the cone
$C_{(x_{u,s},0)}^c$. 
We use lemma \ref{thLipschitzReturnMap}: $\Psi(x_{u,s},0) \in \tilde{\Sigma}$ by invariance and contraction of the cone $C_0^\mathrm{u,s}$.
Thus, by invariance and Expansion of $C_{(x_{u,s},0)}^c$, 
the image of the straight line $\{x_{u,s}\} \times [-\delta,\delta]$ under $\Psi$ contains a curve in $C_{\Psi(x_{u,s},0)}^\mathrm{c}$ 
connecting the extremal plains $\{x_c=\pm\delta\}$. By intermediate-value theorem this curve must intersect the graph of $\zeta$.

(ii) Let $\zeta\in X$ and $x_{u,s},\tilde{x}_{u,s} \in [0,\delta]^2\setminus\{0\}$ be given with 
$\left(\Psi^{-1}(x_{u,s},\zeta(x_{u,s}))\right)_{u,s} = \left(\Psi^{-1}(\tilde{x}_{u,s},\zeta(\tilde{x}_{u,s}))\right)_{u,s} \in [0,\delta]^2$. 
Note in particular that $\left(\Psi^{-1}(x_{u,s},\zeta(x_{u,s}))\right)_u \ne 0$, as the singular line is mapped onto the origin.

Thus $\Psi^{-1}(\tilde{x}_{u,s},\zeta(\tilde{x}_{u,s})) \in C_{\Psi^{-1}(x_{u,s},\zeta(x_{u,s}))}^\mathrm{c}$ and by cone invariance 
$(\tilde{x}_{u,s},\zeta(\tilde{x}_{u,s})) \in C_{(x_{u,s},\zeta(x_{u,s}))}^\mathrm{c}$.
The Lipschitz-bound on $\zeta\in X$ on the other hand implies 
$(\tilde{x}_{u,s},\zeta(\tilde{x}_{u,s})) \in C_{(x_{u,s},\zeta(x_{u,s}))}^\mathrm{s,u}$,
thus $(\tilde{x}_{u,s},\zeta(\tilde{x}_{u,s})) = (x_{u,s},\zeta(x_{u,s}))$.

(iii) Again, the Lipschitz-bound on $\zeta\in X$ translates to 
$(\tilde{x}_{u,s}, \zeta(\tilde{x}_{u,s})) \in C_{(x_{u,s}, \zeta(x_{u,s}))}^\mathrm{u,s}$ for all $x,\tilde{x}$.
Cone invariance, lemma \ref{thLipschitzReturnMap}, immediately yields the Lipschitz bound on $G\zeta$.

(iv) The singular line is fixed by construction, thus we only have to estimate the distance of the nonsingular part.
Let $\zeta,\tilde{\zeta}\in X$ and $x_{u,s},\tilde{x}_{u,s} \in [0,\delta]^2\setminus\{0\}$ be given with 
$\left(\Psi^{-1}(x_{u,s},\zeta(x_{u,s}))\right)_{u,s} = \left(\Psi^{-1}(\tilde{x}_{u,s},\tilde{\zeta}(\tilde{x}_{u,s}))\right)_{u,s} \in [0,\delta]^2$.

Again, this implies $\Psi^{-1}(\tilde{x}_{u,s},\tilde{\zeta}(\tilde{x}_{u,s})) \in C_{\Psi^{-1}(x_{u,s},\zeta(x_{u,s}))}^\mathrm{c}$ 
and by cone invariance we have $(\tilde{x}_{u,s},\tilde{\zeta}(\tilde{x}_{u,s})) \in C_{(x_{u,s},\zeta(x_{u,s}))}^\mathrm{c}$.
Thus we can estimate
\[
\begin{array}{rcl}
|\tilde{\zeta}(\tilde{x}_{u,s}) - \zeta(x_{u,s})| 
   &\le& \|\tilde{\zeta} - \zeta\|_\mathrm{sup} + \sigma\|\tilde{x}_{u,s} - x_{u,s}\|
\\ &\le& \|\tilde{\zeta} - \zeta\|_\mathrm{sup} + \sigma^2 |\tilde{\zeta}(\tilde{x}_{u,s}) - \zeta(x_{u,s})|
\end{array}
\]
The first inequality uses the Lipschitz bound on $\zeta\in X$ whereas the second one uses the cone $C_{(x_{u,s},\zeta(x_{u,s}))}^\mathrm{c}$.
We obtain
\[
|\tilde{\zeta}(\tilde{x}_{u,s}) - \zeta(x_{u,s})|  \;\le\;  \frac{1}{1-\sigma^2} \|\tilde{\zeta} - \zeta\|_\mathrm{sup}.
\]
On the other hand, the expansion of $C_{\Psi^{-1}(x_{u,s},\zeta(x_{u,s}))}^\mathrm{c}$ under $\Psi$ yields
\[
\begin{array}{rcl}
&& \hspace{-5em} 
\left| (G\tilde{\zeta} - G\zeta) \left(\left(\Psi^{-1}(x_{u,s},\zeta(x_{u,s}))\right)_{u,s}\right) \right|
\\ &=& 
\left| \left(\Psi^{-1}(x_{u,s},\zeta(x_{u,s}))\right)_c - \left(\Psi^{-1}(\tilde{x}_{u,s},\tilde{\zeta}(\tilde{x}_{u,s}))\right)_c \right|
\\ &\le&
\frac{1}{K_c} \left| \zeta(x_{u,s}) - \tilde{\zeta}(\tilde{x}_{u,s}) \right|
\\ &\le&
\frac{1}{K_c(1-\sigma^2)} \|\tilde{\zeta} - \zeta\|_\mathrm{sup}.
\end{array}
\]
Lemma \ref{thLipschitzReturnMap} provides constants $K_c$, $\sigma$ with $K_c(1-\sigma^2) > 1$. 
Therefore the last estimates yields the claimed contraction, $\kappa = 1/(K_c(1-\sigma^2))$, and finishes the proof.
\end{proof}

The stable set given by the theorem determines the initial conditions of trajectories that possess the
same $\alpha$-limit dynamics as the 3-periodic heteroclinic cycle:
there exists a codimension-one set of trajectories that converge to the heteroclinic cycle as 
they approach the big-bang singularity in backwards time.


\section{Discussion}
\label{secDiscussion}

We have shown that the set of trajectories within the vacuum Bianchi system, that start close to the period-3 heteroclinic cycle and follow it in the $\alpha$-limit $t\to-\infty$, form a codimension-one Lipschitz-manifold.

This holds true in Bianchi-IX as well as Bianchi-VIII domains. 
In fact, the period-3 heteroclinic cycle is a cycle of three pairs of heteroclinic orbits and the 
domain depends on the choice of representatives.

The same technique is applicable to non-vacuum solutions provided the additional eigenvalue $\mu_\Omega$ 
of the Kasner equilibrium is stronger than the unstable eigenvalue, see (\ref{eqKasnerLinearization}) in reversed time direction.
In fact the analysis of the local map in section \ref{secLocalMap} remains valid for arbitrary stable dimension and order of stable eigenvalues, 
provided all stable eigenvalues are stronger than the unstable one, just the dimension of the component $x_{s,ss}$ changes.
(Note also that the linearization (\ref{eqKasnerLinearization}) is diagonalizable even at points of the Kasner circle 
at which some of the eigenvalues coincide.)
The crucial assumption is a positive gap, $|\mu_s| - |\mu_u| > 0$, 
between the weakest stable and the unstable eigenvalue of the diagonalizable linearization at the Kasner circle.
The result on the stable manifold of the 3-periodic heteroclinic chain 
therefore extends to the non-vacuum Bianchi model provided $3(2-\gamma) > \frac{3}{2}(1-\sqrt{5})$,
i.e.
\[
\gamma < \frac{5-\sqrt{5}}{2},
\]
see (\ref{eqKasnerLinearization}) and note the reversed time.

The result also extends to arbitrary periodic heteroclinic chains, as they keep a uniform distance 
from the Taub points $T_1, T_2, T_3$. To this end it is important to note that all our estimates are
uniform on compact pieces of the Kasner circle which do not contain the Taub points. Similarly,
the Kasner map $\Phi$ is uniformly expanding on such compacta. 
Therefore, the return map associated to an arbitrary periodic chain is composed of finitely many local 
and global maps as studied in this paper.
This yields uniform estimates for each of the local and global maps. Hence, the graph transformation 
of theorem \ref{thStableSet} can be applied to the concatenation of these local return and global excursion maps.
As above we get a codimension one invariant Lipschitz manifold of points that approach the periodic chain
of heteroclinic orbits.
Note that these chains may even include heteroclinic orbits, that approach the Kasner circle in the direction 
corresponding to the weak eigenvalue.
The results on the local passage remain valid for in-sections $\{x_s = 1\}$.
We used notation consistent with the non-principal direction of the periodic-3 chain, 
however the order of stable eigenvalues was not used in any of the proofs.

The unstable eigenvalue along the Kasner circle $\mathcal{K}$ in reversed time direction is bounded by $|\mu_u| < 2$, see (\ref{eqKasnerLinearization}) in reversed time direction.
The extension of the result on arbitrary periodic chains to the non-vacuum Bianchi model is therefore valid as long as the eigenvalue $\mu_\Omega$ to the eigenvector transverse to the vacuum boundary is stronger than 2, i.e. 
\[
\gamma < \frac{4}{3},
\]
including all matter models between dust ($\gamma=1$, included) and radiation ($\gamma=4/3$, excluded).

Non-periodic heteroclinic sequences can be treated in the same way, as long as they keep a uniform distance 
from the Taub points $T_1, T_2, T_3$. Thus the limit of the analysis presented here is the following result:

\begin{rem}
Let any heteroclinic chain of Bianchi-II vacuum solutions, i.e.~a sequence $S = (\Phi^n(q))_{n\ge0}$ of the Kasner map (\ref{eqKasnerMap}), be given 
that keeps a uniform distance from the Taub points $T_1, T_2, T_3$.

Then there exists a local codimension-1 Lipschitz manifold $W^u(S)$ of initial conditions with backward trajectories in the Bianchi system (\ref{eq5dimBianchi}) 
following the given chain $S$.

More precisely, take such an initial condition, say $x_0\in W^u(S)$, and its backward trajectory $\{x_0\cdot t, \; t \le 0\}$. 
Then there exists a sequence $(t_n)_{n\ge0}$, $t_n < 0$, $t_n \searrow -\infty$, such that the distance of $x_0 \cdot t_n$ and $\Phi^n(q)$ goes to zero as $n\to\infty$.
Additionally the pieces $\{x_0\cdot t, \; t_{n+1} \le t \le t_n\}$ of the trajectory converge to the respective heteroclinic orbits connecting $\Phi^{n+1}(q)$ and $\Phi^{n}(q)$, e.g. with respect to Hausdorff distance.
\end{rem}

Note that each transition in the sequence $S$ corresponds to a pair of heteroclinic orbits contained in one of the 
three Bianchi class II ellipsoids. 
Fixing the sign of the nonzero variable $N_k$ in each of the three ellipsoids associates a unique heteroclinic chain 
to the Kasner sequence. Three equal signs yield a manifold in Bianchi class IX, 
the other choices correspond to Bianchi class VIII.
The result holds in all cases and gives 8 local Lipschitz manifold with boundaries in Bianchi classes VI$_0$, VII$_0$.
These boundaries correspond to solutions not following the chain but converging to a Kasner equilibrium after one transition.
The union of the 8 pieces yields 2 local Lipschitz manifolds without boundaries 
but crossing the subspaces of Bianchi classes VI$_0$, VII$_0$.
The manifolds are attached to the 2 heteroclinic orbits given by the first transition of the sequence $S$.

Independently, in \cite{Beguin2010-BianchiAsymptotics} similar results under additional non-resonance conditions 
of the eigenvalues of the equilibria along the Kasner sequence are obtained by linearization techniques.
Remarkably, these non-resonance conditions exclude in particular all periodic sequences.

The major question that still remains open in our treatment is the dynamics corresponding to heteroclinic 
sequences that come arbitrarily close to the Taub points $T_1, T_2, T_3$.
In fact, generic sequences fall into this last category. Here, the very technical work of Reiterer and Trubowitz \cite{ReitererTrubowitz2010-BKL} provides positive results at least in the vacuum case. 

Nevertheless, the question whether the Kasner map already determines the \emph{complete}
dynamics of nearby solutions is still open. The answer will probably require a more detailed study of the 
delicate behavior close to the Taub points.

\noindent
\textbf{Acknowledgement:} The first author was partially supported by the 
Collaborative Research Center 647 ``Space---Time---Matter'' of the German Research Foundation (DFG).
The authors thank the anonymous referee for several suggestions which helped to improve the presentation 
of this work.


\bibliographystyle{alpha}
\bibliography{LieHaeWebGeo-Bianchi3Cycle}

\end{document}